\documentclass{llncs}

\usepackage{amsmath,amssymb,amsthm}

\newtheorem{observation}[theorem]{Observation}

\spnewtheorem{claim}[theorem]{Claim}{\itshape}{\itshape}

\usepackage{textcomp}

\usepackage[utf8]{inputenc}

\usepackage[arxiv_version]{optional}
\usepackage[version=3]{mhchem}
\usepackage{commands-tam}
\usepackage{graphicx}
\usepackage{url}
\usepackage{wrapfig}
\usepackage{cite}
\usepackage{subfigure}
\usepackage{marginnote}				
\usepackage[pdfstartview=FitH,colorlinks=true,citecolor=blue,linkcolor=blue]{hyperref}

\usepackage[usenames]{color}
\usepackage{comment}

\usepackage[disable,textsize=scriptsize]{todonotes} 

\usepackage[normalem]{ulem} 

\makeatletter
\newcommand{\oset}[3][0ex]{%
  \mathrel{\mathop{#3}\limits^{
    \vbox to#1{\kern-2\ex@
    \hbox{$\scriptstyle#2$}\vss}}}}
\makeatother

\renewcommand{\vec}[1]{{\ensuremath{\mathbf{#1}}}}
\newcommand{\vecofvec}[1]{{\ensuremath{\oset[-0.2ex]{\rightharpoonup}{\mathbf{#1}}}}}

\newcommand{\vvc}{\vecofvec{c}}

\newcommand{\vx}{\vec{x}}

\newcommand{\vm}{\vec{m}}
\newcommand{\vd}{\vec{d}}
\newcommand{\vi}{\vec{i}}

\newcommand{\vo}{\vec{o}}

\renewcommand{\vb}{\vec{b}}
\newcommand{\vc}{\vec{c}}

\newcommand{\vv}{\vec{v}}
\newcommand{\vh}{\vec{h}}

\newcommand{\vI}{\vec{I}}
\newcommand{\vA}{\vec{A}}
\newcommand{\vM}{\vec{M}}

\newcommand{\ints}[1]{\{1,\ldots,#1\}}

\usepackage[font={footnotesize}]{caption}

\newcommand{\calM}{\mathcal{M}}

\newcommand{\amax}{\mathrm{amax}}

\newenvironment{subproof}[1][\proofname]{%
  \begin{proof}[#1]%
}{%
  \end{proof}%
}

\title{\vspace{-30pt}Thermodynamic Binding Networks 
}
\author{
David Doty\inst{1}\thanks{doty@ucdavis.edu. Supported by NSF grant CCF-1619343.}
\and 
Trent A.~Rogers\inst{2}\thanks{tar003@email.uark.edu Supported by the NSF Graduate Research Fellowship Program under Grant No. DGE-1450079, NSF Grant CAREER-1553166, and NSF Grant CCF-1422152.} 
\and 
David Soloveichik\inst{3}\thanks{david.soloveichik@utexas.edu. Supported by NSF grants CCF-1618895 and CCF-1652824.}
\and 
Chris Thachuk\inst{4}\thanks{thachuk@caltech.edu.  Supported by NSF grant CCF-1317694.} 
\and 
Damien Woods\inst{5}\thanks{damien.woods@inria.fr. Part of this work was carried out at California Institute of Technology. Supported by Inria (France) as well as National Science Foundation (USA) grants CCF-1219274, CCF-1162589, CCF-1317694.
} 
}
\date{}

\institute{%
\textsuperscript{1} University of California, Davis,   \textsuperscript{2}University of Arkansas,  \textsuperscript{3}University of Texas at Austin, 
\textsuperscript{3}California Institute of Technology, 
\textsuperscript{5}Inria
}

\usepackage{graphicx}

\begin{document}

\pagestyle{plain}

\maketitle

\begin{abstract}
    Strand displacement and tile assembly systems are designed to follow prescribed kinetic rules (i.e., exhibit a specific time-evolution). However, the expected behavior in the limit of infinite time---known as thermodynamic equilibrium---is often incompatible with the desired computation. Basic physical chemistry implicates this inconsistency as a source of unavoidable error. Can the thermodynamic equilibrium be made consistent with the desired computational pathway? In order to formally study this question, we introduce a new model of molecular computing in which computation is driven by the thermodynamic driving forces of enthalpy and entropy. 
    To ensure greatest generality we do not assume that there are any constraints imposed by geometry and treat monomers as unstructured collections of binding sites.
    In this model we design Boolean AND/OR formulas, as well as a self-assembling binary counter, where the thermodynamically favored states are exactly the desired final output configurations.
    Though inspired by DNA nanotechnology, the model is sufficiently general to apply to a wide variety of chemical systems.
\end{abstract}

\section{Introduction}
Most of the models of computing that have come to prominence in molecular programming are essentially kinetic.
For example, models of DNA strand displacement cascades and algorithmic tile assembly formalize desired interaction rules followed by certain chemical systems over time~\cite{phillips2009programming,winfree1998algorithmic}.
Basing molecular computation on kinetics is not surprising given that computation itself is ordinarily viewed as a \emph{process}.
However, unlike electronic computation, where thermodynamics holds little sway, 
chemical systems operate in a Brownian environment~\cite{bennett1982thermodynamics}.
If the desired output happens to be a meta-stable configuration, then thermodynamic driving forces will inexorably drive the system toward error.
For example, 
\emph{leak} in most strand displacement systems occurs because the thermodynamic equilibrium of a strand displacement cascade favors incorrect over the correct output, or does not discriminate between the two~\cite{thachuk2015leakless}.
In DNA tile assembly, we typically must find and exploit kinetic barriers to unseeded growth to enforce that growth happens only from seed assemblies, otherwise thermodynamically favored assemblies will quickly form that are not the intended self-assembly program execution from the seed/input~\cite{SchWin09,barish2009information}.

We introduce the Thermodynamic Binding Networks (TBN) model, where information processing is due entirely to the thermodynamic tradeoff between \emph{entropy} and \emph{enthalpy}, and not any particular reaction pathway.
In most experimental systems considered in DNA nanotechnology, thermodynamic favorability is determined by a tradeoff between:
(1) the number of base pairs formed or broken (all else being equal, a state with more base pairs bound is more favorable);
(2) the number of separate complexes (all else being equal, a state with more free complexes is more favorable). 
We use the terms enthalpy and entropy to describe (1) and (2) respectively (although this use does not perfectly align with their physical definitions, see Section~\ref{sec:model}). 
Intuitively, the entropic benefit of configurations with more separate complexes is due to additional microstates, each describing the independent three-dimensional positions of each complex.
Although the general case of a quantitative trade-off between enthalpy and entropy is complex, we develop an elegant formulation based on the limiting case in which enthalpy is infinitely more favorable than entropy.
Intuitively, this limit corresponds to increasing the strength of binding, while diluting (increasing the volume), such that the ratio of binding to unbinding rate goes to infinity.
Systems studied in molecular programming can in principle be engineered to arbitrarily approach this limit.
Indeed, this is the regime previously studied in the context of leak reduction for strand displacement cascades~\cite{thachuk2015leakless}.
Figure~\ref{fig:tbn-example} shows a simple TBN, which can exist in 9 possible binding configurations. 
The favored (stable) configuration is the one that, among the maximally bound ones (bottom row), maximizes the number of separate complexes (bottom right).

As a central choice in seeking a general theory, we \emph{dispense with geometry}: 
formally, we treat monomers simply as multisets of binding sites (domains).
Viewed in the context of strand displacement, this abstracts away secondary structure (the order of domains on a strand),
allowing us to represent arbitrary molecular arrangements such as pseudoknots~\cite{dirks2007thermodynamic},
and handle non-local error modes such as spurious remote toeholds~\cite{genot2011remote}.
In the context of tile self-assembly, we consider configurations in which binding does not follow the typical regular lattice structure.
Since the TBN model does not rely on geometric constraints to enforce correct behavior, showing that specific undesired behavior is prevented by enthalpy and entropy alone leads to a stronger guarantee.
Thus, for example proving leaklessness in this model would imply that even if pseudoknots, or other typically disallowed structures form, we would still have little leak.
Indeed, by casting aside the vagaries of DNA biophysics (e.g., persistence length, number of bases per turn, sequence dependence on binding strength, etc.), our aim is to develop a general theory of programmable systems based on molecular bonds, a theory that will apply to bonds based on other substrates such as proteins, base stacking, or electric charge.

After introducing the TBN model in Section~\ref{sec:model}, we give results on Boolean circuit-based and self-assembly-based computation. 
In Section~\ref{sec:boolean} we show how to construct AND and OR gates where the thermodynamically favored configurations encode the output.
We develop provable guarantees on the entropic penalty that must be overcome to produce an incorrect $1$ output,
showing how the logic gates can be designed to make the penalty arbitrarily large.
Although completely modular reasoning seems particularly tough in this model, we develop a proof technique based on logically excising domains to handle the composition of Boolean gates---specifically trees of AND gates. 
Further work is needed to generalize these results to arbitrary circuits.

In Section~\ref{sec:polymer-bounds} we look at self-assembly, beginning with questions about large assemblies.  
On the one hand we exhibit a class of TBNs with thermodynamicall stable assemblies (with simple `tree' connectivity) of size exponential in the number of constituent monomer types. On the other hand, we show that this bound is essentially tight by giving an exponential size upper bound on the size of stable assemblies in general.
These self-assembly results, along with the binary counter result below, tell us that  {\em monomer-efficient} self-assembly is indeed possible within this model, 
but that (somewhat surprisingly for a model that favors enthalpy infinitely over entropy)
super-exponential size polymers are necessarily unstable, 
even if they are self-assemblable in  kinetic-based models.

For clarity of thought in separating the computational power of thermodynamics and kinetics, throughout much of this paper we do not identify any particular kinetic pathway leading to the desired TBN stable state. Of course real-world physical systems do not operate at thermodynamic equilibrium, and might take longer than the lifetime of the universe to get there. 
Thus, for such `kinetically trapped' systems, encoding desired output in thermodynamic equilibrium is not enough by itself. 
\opt{arxiv_version}{To address this, in Section~\ref{sec:kineticexamples} 
we give a kinetically and thermodynamically favoured binary counter that assembles in both the abstract Tile Assembly Model and the TBN model.}
\opt{LNCS_version}{To address this, in the full version of this paper we give a kinetically and thermodynamically favoured binary counter that assembles in both the abstract Tile Assembly Model and the TBN model.}
Similarly, the strand displacement AND gate from ref.~\cite{thachuk2015leakless} can be shown to compute correctly in the TBN model~\cite{keenan}.
Nonetheless, more work is needed to come up with TBN schemes that have fast kinetic pathways, in addition to the provable thermodynamic guarantees.

\section{Model}\label{sec:model}

Let $\N,\Z,\Z^+$ denote the set of nonnegative integers, integers, and positive integers, respectively.
A key type of object in our definitions is a multiset, which we define in a few different ways as convenient. 
Let $\calA$ be a finite set.
We can define a multiset over $\calA$ using the standard set notion, e.g., $\vc = \{a, a, c\}$, where $a, c \in \calA$.
Formally, we view multiset $\vc$ as a vector assigning counts to $\calA$.
Letting $\N^\calA$ denote the set of functions $f:\calA\to\N$,
we have $\vc \in \N^\calA$. 
We index entries by elements of $a \in \calA$, calling $\vc(a) \in \N$ the \emph{count of $a$ in $\vc$}.
Fixing some arbitrary ordering on the elements of $\calA = \{a_1,a_2,\ldots,a_{k}\}$, we may equivalently view $\vc$ as an element of $\N^k$, where for $i \in \{1,2,\ldots,k\}$, $\vc(i)$ denotes $\vc(a_i)$.
Let $\|\vc\| = \|\vc\|_1 = \sum_{a\in \calA} \vc(a)$ denote the \emph{size} of $\vc$.
For any vector or matrix $\vc$, let $\amax(\vc)$ denote the largest absolute value of any component of $\vc$.

We model molecular bonds with precise binding specificity abstractly as binding ``domains'',
designed to bind only to other, specific binding domains.
Formally, consider a finite set $\calD$ of \emph{primary domain types}.
Each primary domain type $a \in \calD$ is mapped to a \emph{complementary domain type} (a.k.a., \emph{codomain type}) denoted $a^*$.
Let $\calD^* = \{a^* \mid a \in \calD\}$ denote the set of codomain types of $\calD$.
The mapping is assumed 1-1, so $|\calD^*| = |\calD|$.
We assume that a domain of primary type $a \in \calD$ binds only to its corresponding complementary type $a^* \in \calD^*$, and vice versa.\footnote{That is, we assume \emph{like-unlike} binding such as that found in DNA Watson-Crick base-pairing, as opposed to \emph{like-like} binding such as hydrophobic molecules with an affinity for each other in aqueous solution, or base stacking between the blunt ends of DNA helices\cite{sungwook_bonds_2011, dietz_bonds_science_2015}.
It is not clear the extent to which this choice affects the computational power of our model.
}
The set $\calD \cup \calD^*$ is the set of \emph{domain types}.

We assume a finite set $\calM$ of \emph{monomer types}, where a monomer type $\vm \in \N^{\calD \cup \calD^*}$ is a non-empty multiset of domain types, e.g., $\vm = \{a, b, b, c^*, a^*\}$ with  $a, b, c \in \calD$ being primary domain types. 
A  \emph{thermodynamic binding network} (TBN) is a pair $\calT = (\calD,\calM)$ consisting of a finite set $\calD$ of primary domain types and a finite set $\calM \subset \N^{\calD \cup \calD^*}$ of monomer types.
A \emph{monomer collection} $\vvc \in \N^\calM$ of $\calT$ is multiset of monomer types;
intuitively, $\vvc$ indicates how many of each monomer type from $\calM$ there are, but not how they are bound.\footnote{Because a monomer collection is a multiset of monomer types, each of which is itself a multiset,
we distinguish them typographically with an arrow.}

\begin{figure}[t]
\centering
\includegraphics[width=0.6\textwidth]{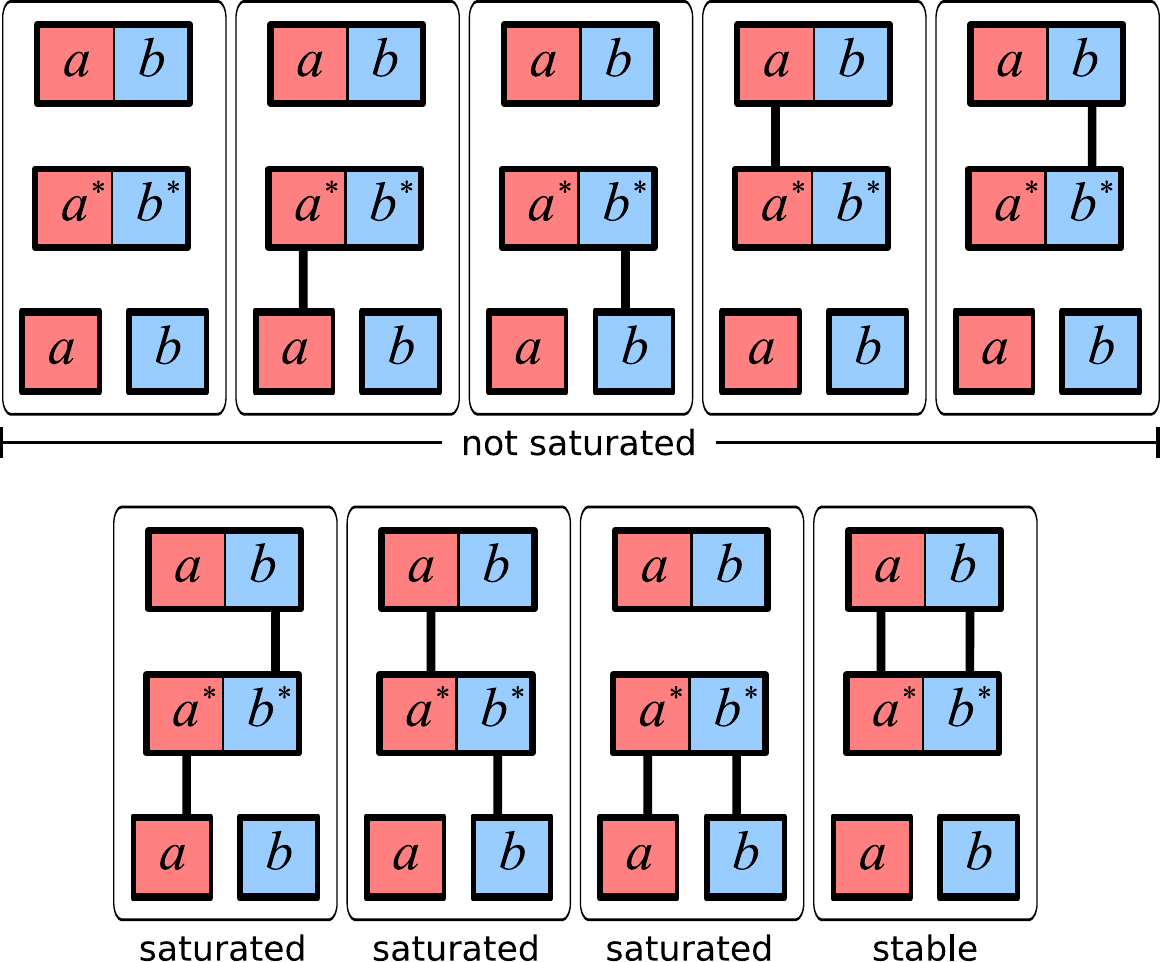}
\caption{
An example TBN $\calT = (\calD,\calM)$.
$\calD = \{a, b\}$ and
$\calM = \{\vm_1, \vm_2, \vm_3, \vm_4\}$,
where monomers $\vm_1=\{a,b\}, \vm_2=\{a^*, b^*\}, \vm_3=\{a\}$, and $\vm_4=\{b\}$.
Note that the order of domains does not matter (thus, $\{a,b\} = \{b,a\}$).
There are nine distinct configurations for the monomer collection $\vvc = \{\vm_1, \vm_2, \vm_3, \vm_4\}$
consisting of a single copy of each of these monomers.
The five in the top row are \emph{not saturated} meaning that 
they do not mazimize the number of bound domains, whereas the four configurations in the bottom row are all \emph{saturated}.
In addition to being saturated, the configuration in the bottom right is \emph{stable} as it maximizes the number of separate complexes (3) among all saturated configurations (the other saturated configurations have 2).  
}
\label{fig:tbn-example}
\end{figure}

Since one monomer collection usually contains more than one copy of the same domain type,
we use the term \emph{domain} to refer to each copy separately.\footnote{For instance, the monomer collection shown in  Fig.~\ref{fig:tbn-example} has 2 domains of type $a$, 2 domains of type $b$, and 1 domain of type $a^*$ and $b^*$ each.}
We similarly reserve the term \emph{monomer} to refer to a particular instance of a monomer type if a monomer collection has multiple copies of the same monomer type.

A single monomer collection $\vvc$ can take on different configurations depending on how domains in monomers are bound to each other.
To formally model configurations, we first need the notion of a bond assignment.
 Let $(U,V,E)$ be the bipartite graph  describing all possible bonds, where
$U$ is the multiset of all primary domains in all monomers in $\vvc$,
$V$ is the multiset of all codomains in all monomers in $\vvc$,
and $E$ is the set of edges between primary domains and their complementary codomains $\{ \{u,v\} \mid u \in U, v \in V, v = u^*\}$.
A \emph{bond assignment} $M$ is  a matching\footnote{A matching of a graph is a subset of edges that share no vertices in common. In our case this enforces that a domain is bound to at most one other domain.} on  $(U,V,E)$.
Then, a \emph{configuration} $\alpha$ of monomer collection~$\vvc$ is  the (multi)graph 
$(U \cup V,E_M)$, where 
the edges $E_M$ describe both the association of domains within the same monomer, and the bonding due to $M$. 
Specifically, for each pair of domains $d_i,d_j \in \calD \cup \calD^*$ that are part of the same monomer in $\vvc$, let $\{d_i,d_j\} \in E_M$, calling this a \emph{monomer edge}, and for each edge $\{d_i,d_i^*\}$ in the bond assignment $M$, let $\{d_i,d_i^*\} \in E_M$, calling this a \emph{binding edge}.
Let $[\vvc]$ be the set of all configurations of a monomer collection $\vvc$.
We say the size of a configuration, written $|\alpha|$, is simply the number of monomers in it.

Another graph that will be useful in describing the connectivity of the monomers, independent of which exact domains are bound, is the \emph{monomer binding graph} $G_\alpha = (V_\alpha, E_\alpha)$, which is obtained by contracting each monomer edge of $\alpha$.
In other words, $V_\alpha$ is the set of monomers in $\alpha$, with an edge between monomers that share at least one pair of bound domains.

Which configurations are thermodynamically favored over others depends on two properties of a configuration: its bond count and entropy.
The \emph{enthalpy} $H(\alpha)$ of a configuration is the number\footnote{We are assuming bonds are of equal strength (although
the definition can be naturally generalized to bonds of different strength).} of binding edges (i.e., the cardinality of the matching $M$).
The \emph{entropy} $S(\alpha)$ of a configuration is the number of connected components of $\alpha$.\footnote{Our use of the terms ``enthalpy'' and ``entropy'', and notation $H$ and $S$ is meant to evoke the corresponding physical notions.  
Note, however, that there are other contributions to physical entropy besides the number of separate complexes. 
Indeed, the free energy contribution of forming additional bonds typically contains substantial enthalpic and entropic parts.
}
Each connected component is called a \emph{polymer}.\footnote{We are generalizing the convention for the word ``polymer'' in the chemistry literature.
We have no requirement that a polymer be linear, nor that it consist of repeated subunits.
We chose ``polymer'' rather than ``complex'' to better contrast with ``monomer''.
}
Note that a polymer is itself a configuration, but of a smaller monomer collection $\vvc' \leq \vvc$. 
As with all configurations, the size of a polymer is the number of monomers in~it.

\begin{figure}[t]
\centering
\includegraphics[width=.8\textwidth]{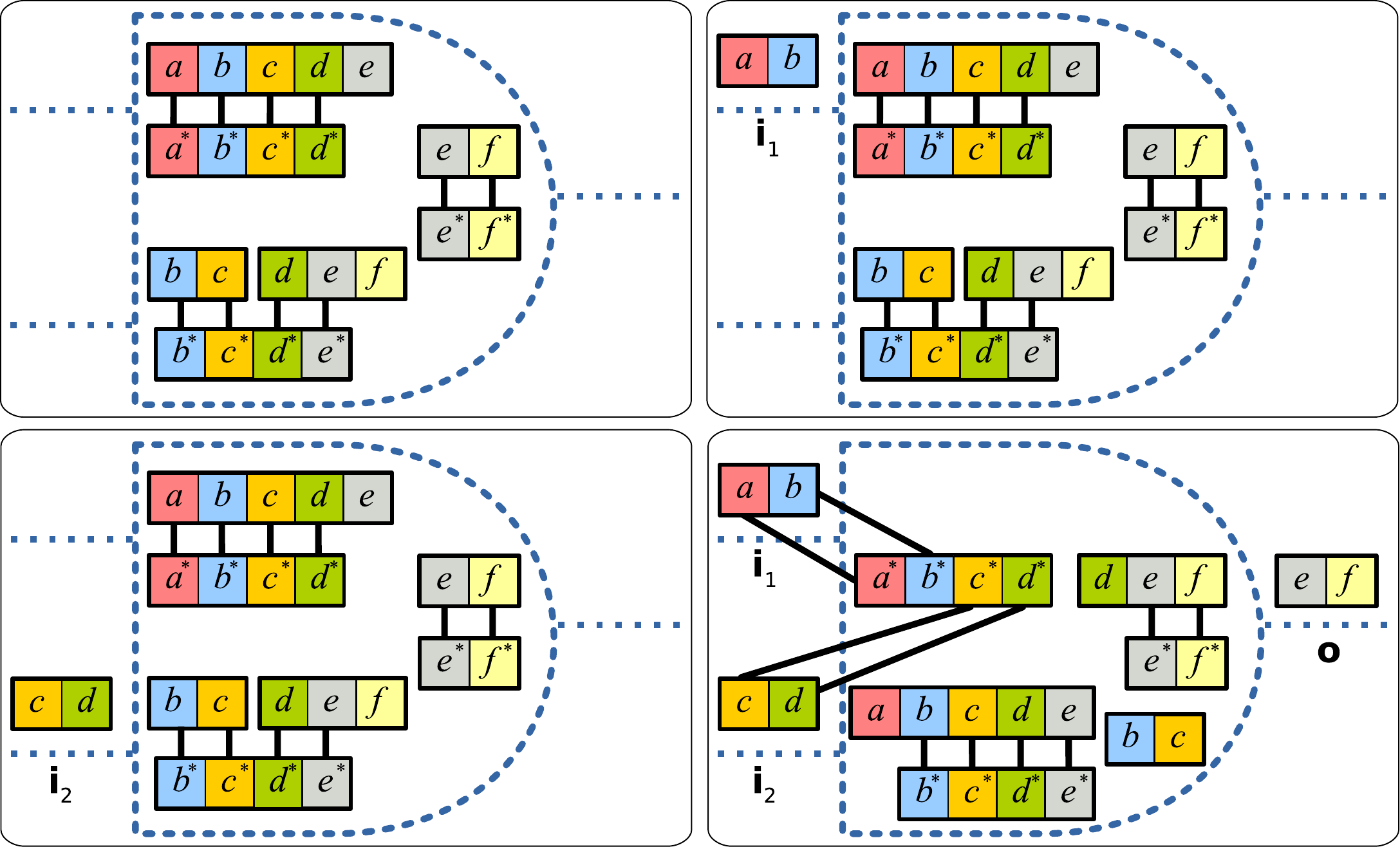}
\caption{
{\bf Basic AND gate:}
Monomers $\vi_1=\{a, b\}$ and $\vi_2=\{c, d\}$ represent the input, $\vo=\{e,f\}$ represents the output, while the remainder 
are intermediate monomers to implement the logic relating the input to the output.
If either or both inputs are missing, then the only stable configuration
has the present input monomers free (unbound) and the output monomer $\vo$ not free (bound).
If both input monomers are present, then there are two stable configurations: one with inputs free (and $\vo$ not free), or the one depicted with $\vo$ free and both inputs bound.
}
\label{fig:simpleANDgates}
\end{figure}

Intuitively, 
configurations with higher enthalpy $H(\alpha)$ (more bonds formed) and higher entropy $S(\alpha)$ (more separate complexes) are thermodynamically favored. 
What happens if there is a conflict between the two?
One can imagine capturing a tradeoff between enthalpy and entropy by some linear combination of $H(\alpha)$ and $S(\alpha)$.
In DNA nanotechnology applications,
the tradeoff can be controlled by increasing the number of nucleotides constituting a binding domain (increasing the weight on $H(\alpha)$), or by decreasing concentration (increasing the weight on $S(\alpha)$).%
\footnote{In typical DNA nanotechnology applications, 
the Gibbs free energy $\Delta G(\alpha)$ of a configuration $\alpha$ can be  estimated as follows.
Bonds correspond to domains of length~$l$ bases,
and forming each base pair is favorable by $\Delta G^\circ_{\text{bp}}$.
Thus, the contribution of $H(\alpha)$ to $\Delta G(\alpha)$ is $(\Delta G^\circ_{\text{bp}} \cdot l) H(\alpha)$.
At $1$ M, the free energy penalty due to decreasing the number of separate complexes by $1$ is $\Delta G^\circ_{\text{assoc}}$.
At effective concentration $C$ M, this penalty increases to 
$\Delta G^\circ_{\text{assoc}} + RT \ln(1/C)$.
As the point of zero free energy, we take the configuration with no bonds, and all monomers separate.
Thus, the contribution of $S(\alpha)$ to $\Delta G(\alpha)$ is $(\Delta G^\circ_{\text{assoc}} + RT \ln(1/C))(|\alpha|-S(\alpha))$, where $|\alpha|$ is the total number of monomers.
To summarize,
\begin{align*}
\Delta G(\alpha) &= 
(\Delta G^\circ_{\text{bp}} \cdot l) H(\alpha) + 
(\Delta G^\circ_{\text{assoc}} + RT \ln(1/C))(|\alpha|-S(\alpha)).
\end{align*}
Note that, as expected, this is a linear combination of $H(\alpha)$ and $S(\alpha)$, and that increasing the length of domains~$l$ weighs $H(\alpha)$ more heavily, while decreasing the concentration $C$ weighs $S(\alpha)$ more heavily.
Typically $G^\circ_{\text{bp}} \approx -1.5$ kcal/mol, and $G^\circ_{\text{assoc}} \approx 1.96$ kcal/mol~\cite{santalucia2004thermodynamics}.
}

In the rest of this paper, we study the particularly interesting limiting case in which enthalpy is \emph{infinitely} more favorable than entropy.\footnote{Note that the other limiting case, where entropy is infinitely more favorable, is degenerate: the most favorable configuration in that case always has every monomer unconnected to any other.}
We say a configuration $\alpha$ is \emph{saturated} if it has no pair of domains $d$ and $d^*$ that are both unbound;
this is equivalent to stating that $\alpha$ has maximal bonding among all configurations in $[\vvc]$.
We say a configuration $\alpha\in[\vvc]$ is \emph{stable} (aka thermodynamically favored) if it is saturated and maximizes the entropy among all saturated configurations, i.e., every saturated configuration $\alpha'\in[\vvc]$ obeys $S(\alpha') \leq S(\alpha)$.
Let $[\vvc]_\Box$ denote the set of stable configurations of monomer collection $\vvc$.
See Fig.~\ref{fig:tbn-example} for an example thermodynamic binding network that has a single stable configuration.
We note that, consistent with our model,
in DNA strand displacement cascades ``long'' domains are assumed to always be paired,
and systems can be effectively driven by the formation of more separate complexes~\cite{ZhaTurYurWin07}.



\section{Thermodynamic Boolean formulas} \label{sec:boolean}

Fig.~\ref{fig:simpleANDgates} shows an example of a TBN that performs AND computation, based on the CRN strand displacement gate from ref.~\cite{thachuk2015leakless}.
Realized as a strand displacement system, it has a kinetic pathway taking the untriggered (left) to the triggered (right) configuration.
The inputs are specified by the presence (logical value 1) or absence (logical value 0) of the input monomers $\vi_1$ and $\vi_2$.   
The output convention followed is the following. 
The output is 1 if and only if \emph{some} stable configuration has the output monomer $\vo$ unbound to any other monomer (free).
This can be termed the \emph{weak} output convention.
Alternatively, in the \emph{strong} output convention,
output 1 implies \emph{every} stable configuration has the output monomer $\vo$ free, and output 0 implies  \emph{every} stable configuration has the output monomer $\vo$ bound to some other monomer. 
More complex AND gate designs are compatible with the strong output convention (not shown).


Note that even the weak output convention, 
coupled with a kinetic pathway releasing the output given the correct inputs, 
can be used to argue that:
(1) if the correct inputs are present the output will be produced (via kinetic argument),
(2) if the correct inputs are not present then ultimately little output will be free (thermodynamic argument).
In the context of strand displacement cascades,  TBNs can explore arbitrary structures (pseudoknots, remote toeholds, etc) since we do not impose any ordering on domains in a monomer, nor any geometry. 
This strengthens the conclusion of (2), showing that arbitrary (even unknown) kinetic pathways must lead to a thermodynamic equilibrium with little output.

While individual AND gates can be proven correct with respect to the above output conventions (e.g., through the SAT solver of ref.~\cite{keenan}),
it remains to be shown that these components can be safely composed into arbitrary Boolean circuits.
Note that the input and output monomers have orthogonal binding sites.
This is important for composing AND gates, where the output of one  acts as an input to another.
As is typical for strand displacement logic, OR gates can be trivially created when multiple AND gates have the same output.
Dual-rail AND/OR circuits are sufficient to compute arbitrary Boolean functions without explicit NOT gates.
Nonetheless it is not obvious that the input convention (complete presence or absence of input monomers) matches the output convention (weak or strong).
It is also not clear how statements about the stable configurations of the whole circuit can be made based on the stable configurations of the individual modules.

We now show that correct composition can be proven in certain cases.
Although we believe that the gate shown in Fig.~\ref{fig:simpleANDgates} is composable, the argument below relies on a different construction.
We further consider a restricted case of AND gate formulas (trees).

An important concept in the argument below is the notion of ``distance to stability''.
This refers to the difference between the entropy of the stable configurations and the largest entropy of a saturated configuration with incorrect output.
The larger the distance to stability, the larger the entropy penalty to incorrectly producing the output.
Unlike the simple AND gate from Fig.~\ref{fig:simpleANDgates}, the constructions below can be instantiated to achieve arbitrary desired distance to stability (by increasing the redundancy parameter $n$). 

Many open questions remain. 
Can our techniques be generalized to arbitrary circuits, rather than just trees of AND gates?
Can we prove these results for logic gates that have a corresponding kinetic pathway (like the AND gates in Fig.~\ref{fig:simpleANDgates} which can be instantiated as strand displacement systems)?
Finally, in our Boolean gate constructions, we assume that the monomer collection has exactly one copy of certain monomers. 
It remains open whether these schemes still work if there are many copies of all monomers.

\subsection{Translator cascades} \label{sec:translatorcascade}

We begin with the simplest of circuits, \emph{translator cascades} ($x_1 \rightarrow  x_2 \rightarrow ... \rightarrow x_{k+1}$), which simply propagate signal through $k$ layers when the input signal $x_1$ is present.
Logically a translator gate is simply a repeater gate.
The input is the presence or absence of the input monomer consisting of $n$ copies of domain $x_1$. 
Our analysis below implies that if and only if the input is present, there is a stable configuration with $n$ copies of $x_{k+1}$ domain in the same polymer.
The \emph{terminator gadget} converts this output to the weak output convention defined above (whether or not the monomer consisting of $n$ copies of domain $x_{k+1}$ is free).
The following Lemma shows that we can exactly compute the distance from stability of a translator cascade shown in Fig.~\ref{fig:translator}.
Besides being a ``warm-up'' for AND gate cascades, the Lemma is used in the proof of Theorem~\ref{thm:and-tree-dst-for-leak}.


\begin{figure}[t]
\centering
\includegraphics[width=.8\textwidth]{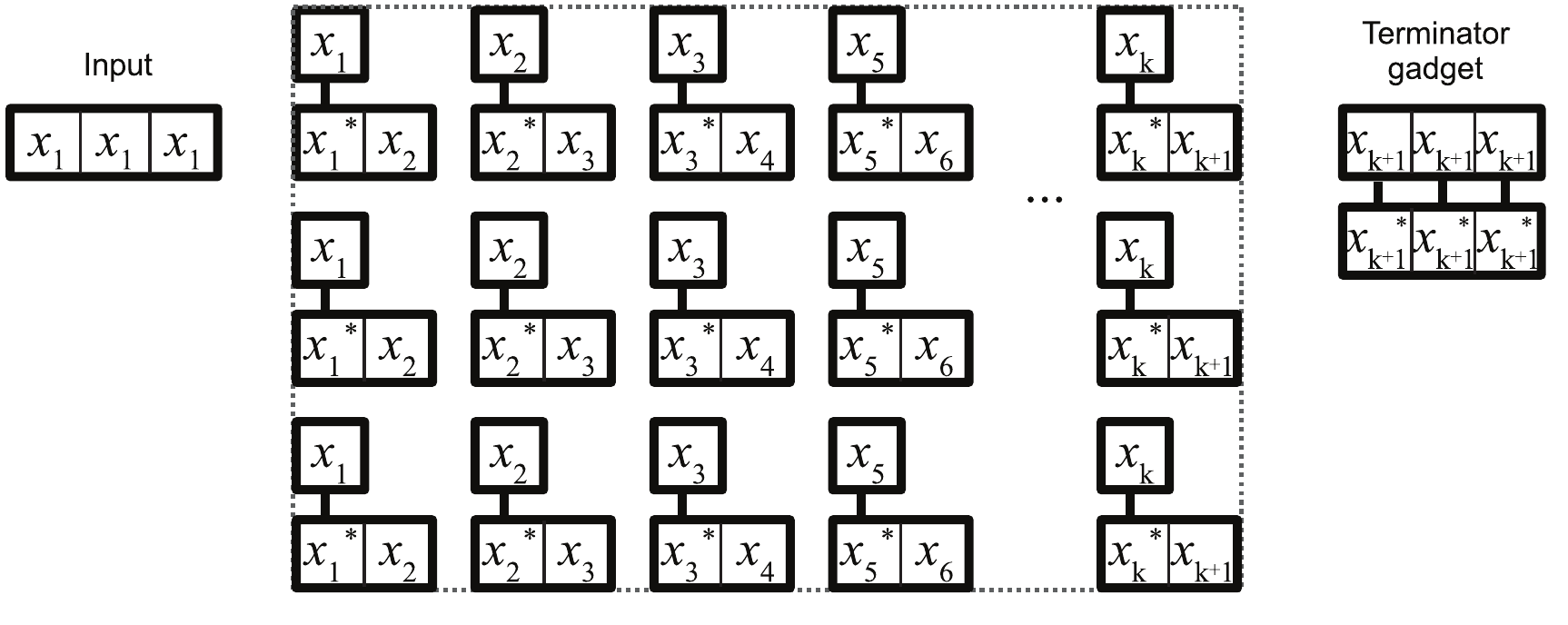}
\caption{
A cascade of $k$ translator gates discussed in Section~\ref{sec:translatorcascade}, with redundancy parameter $n=3$.
We say that a configuration of a formula has output $1$ if the terminator monomer $\{x_{k+1},\ldots,x_{k+1}\}$ is free, and has output $0$ otherwise.
Redundancy parameter $n$ specifies the number of copies of monomers and domains as shown.
}
\label{fig:translator}
\end{figure}


\begin{observation}\label{obs:translator}
  The intended configuration $\alpha$ of a monomer collection representing a depth $k$, redundancy $n$ translator cascade, without input, and with output $0$, is saturated and has $S(\alpha)=nk+1$. (See Fig.~\ref{fig:translator}.)
\end{observation}

\begin{lemma}\label{lem:translator}
  If $\gamma$ is a saturated configuration of a monomer collection representing a depth $k$, redundancy $n$ translator cascade, without input, and with output $1$, then $S(\gamma)=n(k-1)+2$.
\end{lemma}

\opt{arxiv_version}{The proof of Lemma~\ref{lem:translator} appears in Appendix~\ref{app:translator}.}
\opt{LNCS_version}{The proof of Lemma~\ref{lem:translator} appears in the full version of this paper.}
%
Taken together, Observation~\ref{obs:translator} and Lemma~\ref{lem:translator} imply that the redundancy parameter ($n$) guarantees the distance to stability ($n-1$) for a translator cascade of any length.

\subsection{Trees of AND gates}
\label{sec:boolformulas}

In this section we motivate how Boolean logic gates can be composed such that the overall circuit has a guaranteed distance to stability, relative to a redundancy parameter $n$.
Specifically, we start with the AND gate design of Fig.~\ref{fig:AND-gate},
and we give a concrete argument for a tree of these AND gates (e.g., Fig.~\ref{fig:AND-tree-leak-path}).

\begin{figure}[t]
\centering
\includegraphics[width=\textwidth]{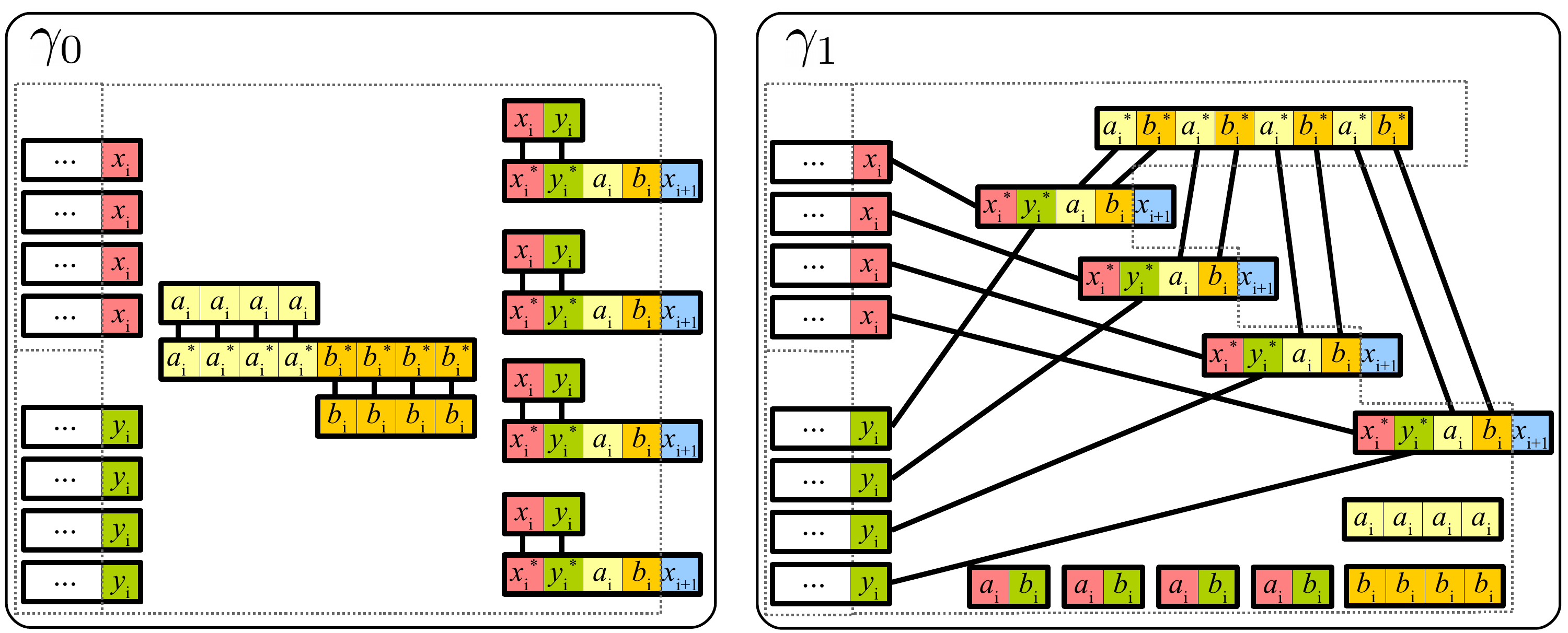}
\caption{AND gates used in Section~\ref{sec:boolformulas}, with redundancy parameter $n=4$.
Two saturated configurations are shown:
$\gamma_0$ is the intended configuration corresponding to output of $0$.
$\gamma_1$ is the intended configuration corresponding to output of $1$.
Input domains are $x_i$  $y_i$, and output domains are $x_{i+1}$.
The output is considered to be 1 in any configuration where all $n$ output domains are in the same 
polymer, 0 otherwise. 
Dashed boxes represent that any domain type appearing inside of a box does have have a complement appearing outside of the box.
}
\label{fig:AND-gate}
\end{figure}

\begin{figure}[t]
\centering
\includegraphics[width=.6\textwidth]{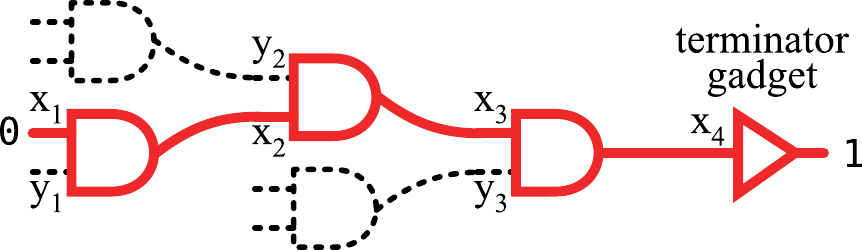}
\caption{Shown highlighted is a leak path through a tree of AND gates from a missing input (``0'') to erroneous output (``1'').}
\label{fig:AND-tree-leak-path}
\end{figure}

\begin{theorem}\label{thm:and-tree-dst-for-leak}
Consider a TBN for AND gates, with redundancy $n$, composed into a tree of depth $k$.
If at least one of the inputs is not present, the distance to stability for any saturated configurations with output $1$ is at least $n-2k-1$.
\end{theorem}

\begin{proof}

Let $\gamma$ be any saturated configuration of the TBN with output $1$.
Consider the missing input and define the \emph{leak path} to be the linear sequence of AND gates from the missing input to and including the terminator gadget.
For convenience we imagine relabelling all the domains in the leak path indexed by the position of the AND gate in the leak path.
For example, Fig.~\ref{fig:AND-tree-leak-path} highlights the leak path through the tree from a missing input (``0'') to erroneous output (``1'').
Specifically, the domain names as shown in Fig.~\ref{fig:AND-gate} appear in the $i$th AND gate (for $1 \leq i \leq k$), where $x_{k+1}$ feeds into the terminator gadget.
Domains $y_i$ connect the leak path to the rest of the tree.

\begin{definition}
Given a configuration $\alpha$ of a monomer collection $\vvc$, we say we \emph{excise} a domain $d$ if we create a new configuration $\alpha'$ by removing the node corresponding to $d$ and all incident edges. (Note that $\alpha'$ is a configuration of a monomer collection of a different TBN.)
\end{definition}

\paragraph{Manipulation 1.} Excise all domains of type $y_i$ and codomains of type $y_i^*$ on monomers of the leak path involved in fan-in,  $1 \leq i \leq k$, yielding the new configuration $\gamma'$. Note that if domain $y_i$ is on a monomer other than the leak path, then it is not excised.

The leak path in $\gamma'$ now has no domains in common with the rest of the tree (and thus no bonds).
Let $\gamma'_L$ be the subconfiguration of the leak path, and let $\gamma'_R$ be the subconfiguration of the rest of the system.  (Note $\gamma’=\gamma'_L \cup \gamma’_R$.)

\begin{observation}\label{obs:excision}
Given a saturated configuration $\alpha$, if you excise all domains or codomains of a particular type (or both its domains and codomains) yielding $\alpha'$, then $\alpha'$ is saturated.
\end{observation}

By Observation~\ref{obs:excision} $\gamma’$ is saturated since for every domain type $y_i$ and codomain type $y_i^*$, every instance of $y_i^*$ is excised; $1 \leq i \leq k$.
This implies $\gamma'_L$ and $\gamma'_R$ are also saturated.

\paragraph{Manipulation 2.}  Excise all domains of type $a_i$ and $b_i$ and all codomains of type $a_i^*$ and $b_i^*$ in $\gamma'_L$, $1 \leq i \leq k$, yielding the new configuration $\gamma''_L$.  By Observation~\ref{obs:excision}, $\gamma''_L$ is saturated.

\begin{claim}\label{clm:manip1}
  $S(\gamma') \geq S(\gamma)$.
\end{claim}

\begin{subproof}[Proof of the claim]
Entropy can only be decreased via excision if an entire monomer is excised.  
Since Manipulation 1 only excised domain and codomain types from the set $\calD'=\bigcup\limits_{i=1}^{k} \{y_i, y_i^*\}$, and those domain types only appear on monomers which also have domain instances with types not in $\calD'$, then no entire monomer was excised.
\end{subproof}

\begin{claim}\label{clm:manip2}
  $S(\gamma''_L) \geq S(\gamma'_L)-3k$.
\end{claim}

\begin{subproof}[Proof of the claim]
For every layer $i$, $1 \leq i \leq k$, there are $3$ monomers that only contain domain and codomain types in the set $\{a_i, b_i, a_i^*, b_i^*\}$.
Therefore, $\gamma''_L$ contains at most $3$ fewer monomers than $\gamma'_L$, for each of the $k$ layers.
\end{subproof}

\begin{claim}\label{clm:entropy-gamma_L}
  $S(\gamma''_L) = n(k-1)+2$.
\end{claim}

\begin{subproof}[Proof of the claim]
Recognize that $\gamma''_L$ is a saturated configuration of a monomer collection representing a depth $k$, redundancy $n$ translator cascade, without input, and with output $1$.
The claim follows by Lemma~\ref{lem:translator}.
\end{subproof}

\begin{claim}\label{clm:entropy-gamma}
  $S(\gamma) \leq n(k-1) + 2 + S(\gamma’_R) + 3k$.
\end{claim}

\begin{subproof}[Proof of the claim]
\begin{align*}
    S(\gamma) &\leq S(\gamma')                        & \text{by Claim~\ref{clm:manip1}}\\
              &= S(\gamma'_L) + S(\gamma'_R)          & \\
              &\leq S(\gamma''_L) + S(\gamma'_R) + 3k & \text{by Claim~\ref{clm:manip2}}\\
              &\leq n(k-1) + 2 + S(\gamma'_R) + 3k    & \text{by Claim~\ref{clm:entropy-gamma_L}}
\end{align*}
\end{subproof}
Now, take the monomers from the leak path in $\gamma$, and configure them into the ``untriggered configuration'' (see Fig.~\ref{fig:AND-gate}, left), yielding subconfiguration $\beta$.
Let $\alpha = \beta \cup \gamma'_R$.
Note that $\beta$ is saturated, and therefore $\alpha$ is a saturated configuration of the entire tree (\textit{i.e.,} the same TBN as $\gamma$).

\begin{observation}\label{obs:alpha-entropy}
$S(\alpha)=S(\gamma'_R) + k(n+1) + 1$.
\end{observation}

Finally, consider the entropy gap between $\alpha$ and $\gamma$.

\begin{align*}
S(\alpha)-S(\gamma) &\geq S(\gamma'_R) + k(n+1) + 1 - S(\gamma) & \text{by Observation~\ref{obs:alpha-entropy}}\\
                    &\geq S(\gamma'_R) + k(n+1) + 1 & \\
                    & \qquad\qquad - \left( n(k-1) + 2 + S(\gamma'_R) + 3k \right) & \text{by Claim~\ref{clm:entropy-gamma}}\\
                    &= n - 2k - 1
\end{align*}

Therefore, there exists a saturated configuration with output $0$ over the same TBN as $\gamma$, but with entropy at least $n - 2k - 1$ larger, thus establishing the theorem.
\end{proof}

Theorem~\ref{thm:and-tree-dst-for-leak} seems to suggest that in order to maintain the bound on distance to stability for incorrect computation, the redundancy parameter $n$ should increase to compensate for an increase in circuit depth $k$.
However, a more sophisticated argument shows that manipulations $1$ and $2$ can decrease entropy by at most $k+1$.
Following the above argument, the distance to stability is found to be $n-2$.
This is optimal because a single AND gate with redundancy $n=2$ can be shown to have no entropy gap between output $0$ and output $1$ configurations.



\section{Thermodynamic self-assembly: Assembling large polymers} \label{sec:polymer-bounds}

TBNs can not only exhibit Boolean circuit computation, but they can also be thought of as a model of self-assembly.
Here we begin to explore this connection by asking a basic question motivated by the abstract Tile Assembly Model (aTAM)~\cite{winfree1998algorithmic}:
how many different monomer \emph{types} are required to assemble a large polymer?

Favoring enthalpy infinitely over entropy, on its face, \emph{appears} to encourage large polymers.
Perhaps we can imagine designing a single TBN $\calT$ that can assemble arbitrarily large polymers where for each $n \in \N$, $\calT$ has a stable polymer~$\alpha$ composed of at least~$n$ monomers.
In this section we show that this is  impossible:
every TBN $\calT = (\calD,\calM)$ has stable polymers of size at most exponential in the number of domain types $|\calD|$ and monomer types $|\calM|$ (Theorem~\ref{thm:polymer-upper-bound}).
The proof shows that any polymer $\rho$ larger than the bound can be partitioned into at least two saturated (maximally bound) polymers, which implies that $\rho$ is not stable.
Fig.~\ref{fig:long-polymer-breaking} gives an example.
We also show that this upper bound is essentially tight by constructing a family of systems with exponentially large stable polymers (Theorem~\ref{thm:polymer-lower-bound-tree}). 
Taken together, the exponential lower bound of Theorem~\ref{thm:polymer-lower-bound-tree} and upper bound of Theorem~\ref{thm:polymer-upper-bound} give a relatively tight bound on the maximum size achievable for stable TBN polymers.

Is it possible to construct algorithmically \emph{interesting} TBN polymers that are stable?
In the full version of this paper, 
we show that a typical binary counter construction from the aTAM model is not stable,
but can be modified to become stable in our model.
Importantly, 
this TBN binary counter demonstrates that in principle algorithmically complex assemblies could have effective assembly pathways (aTAM) as well as be thermodynamically stable (TBN).

\subsection{Superlative trees: TBNs with exponentially large stable polymers}

The next theorem shows that there are stable polymers that are exponentially larger than the number of domain types and monomer types required to assemble them. 

\begin{theorem} \label{thm:polymer-lower-bound-tree}
  For every $n,k \in \Z^+$, there is a TBN $\;\calT = (\calD,\calM)$ with $|\calD|=n-1$ and $|\calM|=n$, having a stable polymer of size $\frac{k^n-1}{k-1}.$
\end{theorem}

\begin{proof}
  An example of $\calT$ for $n=4$ and $k=2$ is shown in Fig.~\ref{fig:tree-polymer}.
  Let $\calD = \{d_1,\ldots,d_n\}$ and $\calM = \{\vm_1,\ldots,\vm_n\}$, where, for each $j \in \{2,\ldots,n-1\}$, $\vm_j = \{d_{j-1}^*, k\cdot d_{j}\}$ (i.e., 1 copy of $d_{j-1}^*$ and $k$ copies of $d_{j}$), $\vm_1 = \{k \cdot d_1\}$, and $\vm_n = \{d_{n-1}^*\}$.
  Define $\vvc\in\N^\calM$ by $\vvc(\vm_j) = k^{j-1}$ for $j \in \ints{n}.$
  Then $\|\vvc\| = \sum_{j=1}^{n} k^{j-1} = \frac{k^n-1}{k-1}$.
  Observe that $[\vvc]$ has a unique (up to isomorphism) saturated configuration $\alpha$ (which is therefore stable), described by a complete $k$-ary tree: level $j \in \ints{n-1}$ of the tree is composed of $k^{j-1}$ copies of $\vm_j$, each bound to $k$ children of type $\vm_{j+1}$ in level $j+1$.
\end{proof}

\begin{figure}[t]
\minipage{0.475\textwidth}
\includegraphics[width=\linewidth]{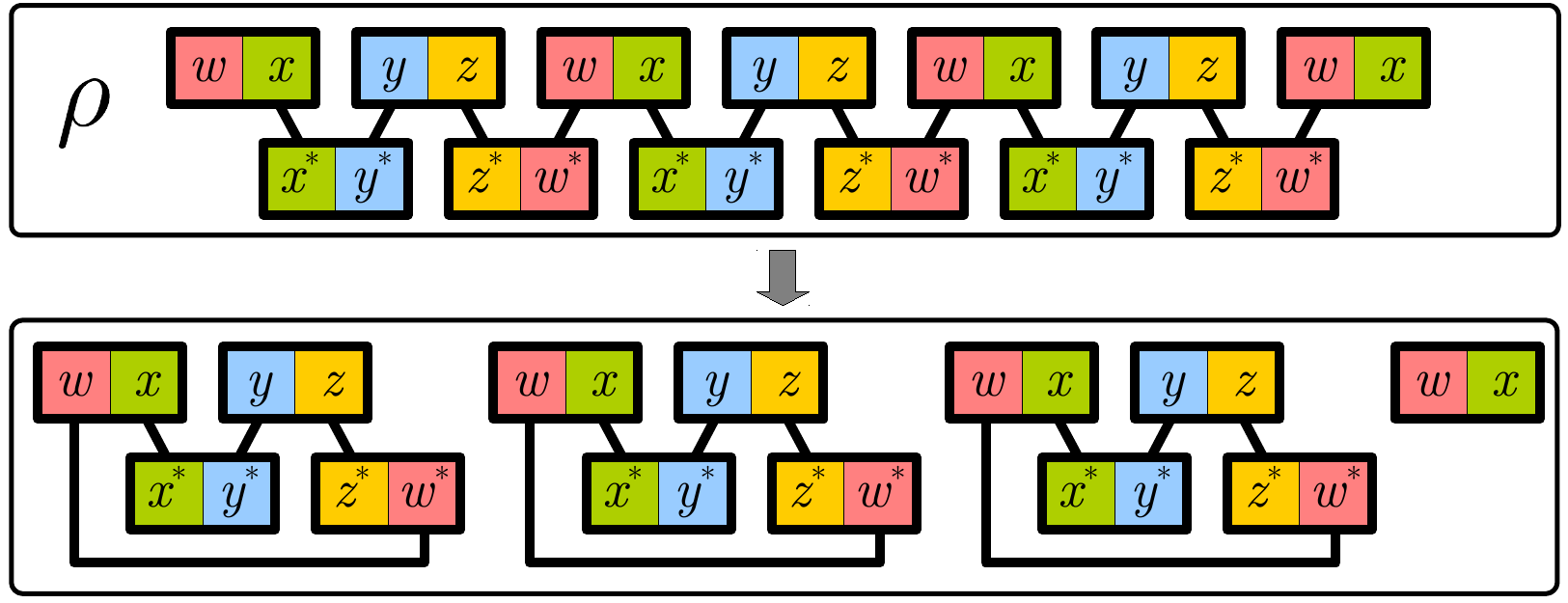}
\caption{
A polymer $\rho$ composed of several copies of four monomer types,
which is not stable since it can be broken into several smaller polymers (bottom panel) such that all domains are bound.
}  
\label{fig:long-polymer-breaking}
\endminipage\hfill
\minipage{0.475\textwidth}
\includegraphics[width=\linewidth]{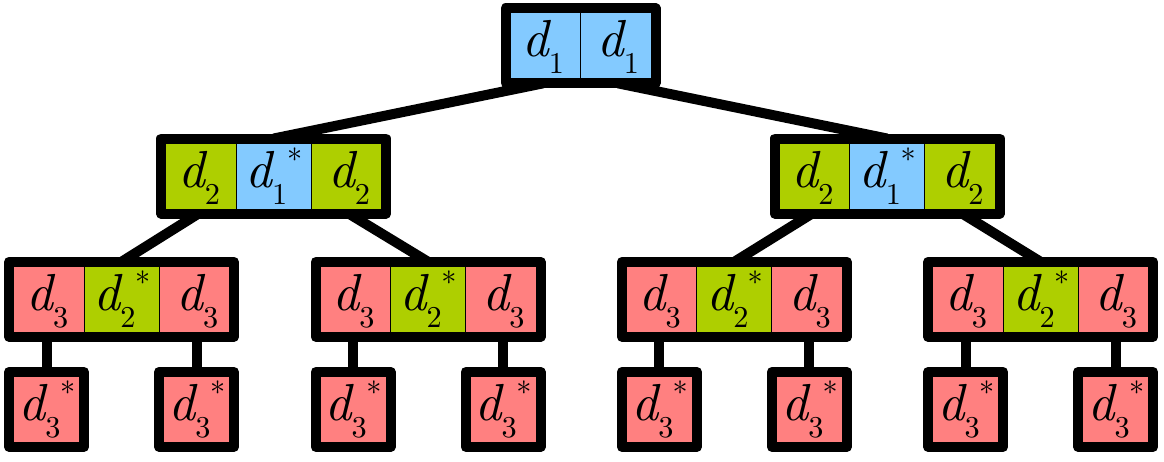}
\caption{
An example of a TBN from Theorem~\ref{thm:polymer-lower-bound-tree} for $n=4$ and $k=2$.
}
\label{fig:tree-polymer}
\endminipage
\end{figure}

The remainder of Section~\ref{sec:polymer-bounds} is devoted to proving that no stable polymer $\rho$ can have size \emph{more} than exponential in $|\calD|$ and $|\calM|$.
\opt{arxiv_version}{We note that this is easy to prove in the special case if the monomer binding graph $G_\rho$ is a tree,
such as the TBN in Theorem~\ref{thm:polymer-lower-bound-tree}  in Appendix~\ref{app:acyclic-self-assembly}.}

\subsection{A linear algebra framework}

We prove Theorem~\ref{thm:polymer-upper-bound},
the main result of Section~\ref{sec:polymer-bounds},
by viewing TBNs from a linear algebra perspective.
Let $\calT=(\calD, \calM)$  be a TBN,
with $\calD = \{d_1,\ldots,d_d\}$
and $\calM = \{\vm_1,\ldots,\vm_m\}$.
For a matrix $\vA$, let $\vA(i,j)$ denote the entry in the $i$'th row and $j$'th column.
Define the $d \times m$ $\emph{positive monomer matrix}$ $\vM_\calT^+$ of $\calT$ by $\vM_\calT^+(i,j) = \vm_j(d_i)$. 
Define the $d \times m$ $\emph{negative monomer matrix}$ $\vM_\calT^-$ of $\calT$ by $\vM_\calT^-(i,j) = \vm_j(d_i^*)$. 
Define the $d \times m$ \emph{monomer matrix} $\vM_\calT$ of $\calT$ to be $\vM_\calT = \vM_\calT^+ - \vM_\calT^-$.
Note that $\vM_\calT^+$ and $\vM_\calT^-$ are matrices over $\N$, but $\vM_\calT$ is over~$\Z$.

The rows of the monomer matrix $\vM_\calT$ correspond to domain types and the columns correspond to monomer types.
The mapping from a TBN $\calT$ to a monomer matrix $\vM_\calT$ is not 1-1:
$\vM_\calT(i,j)$ is the number of $d_i$ domains minus the number of $d_i^*$ domains in monomer type $\vm_j$, which would be the same, for instance, for monomer types $\vm_1 = \{d_1,d_3\}$ and $\vm_2 = \{d_1,d_1,d_1^*,d_3\}$.
Let $\vvc$ be a monomer collection and let $\vd = \vM_\calT \vvc \in \N^d$;
for $i \in \ints{d}$, $\vd(i)$ is the number of $d_i$ domains minus the number of $d_i^*$ domains in the whole monomer collection $\vvc$.

Let $\alpha \in [\vvc]$ be saturated;
$\alpha$ can only have a domain $d_i$ unbound if all copies of its complement $d_i^*$ are bound, and vice versa.
If $\vd(i) > 0$, in $\alpha$
there is an excess of $d_i$ domains, and all $d_i^*$ domains are bound.
If $\vd(i) < 0$, in $\alpha$
there is an excess of $d^*_i$ domains, and all $d_i$ domains are bound.
This leads to the following observation.

\begin{observation} \label{obs:excess-domains-same-across-saturated-configs}
  Let $\calT=(\calD,\calM)$ be a TBN and $\vvc\in\N^\calM$ a monomer collection.
  Let $\vd = \vM_\calT \vvc$.
  Then for every configuration $\alpha \in [\vvc]$, $\alpha$ is saturated if and only if,
  for all $i \in \ints{d}$, if $\vd(i) \geq 0$ (respectively, if $\vd(i) \leq 0$),
  then $\vd(i)$ is the number of unbound $d_i$ (resp., $d_i^*$) domains in $\alpha$.
\end{observation}


Let $\calT = (\calD, \calM)$ and $\calT' = (\calD, \calM')$ be TBNs with the same set of domain types.
Then we call $\calT'$ a \emph{relabeling} of $\calT$ if there exists a subset $D \subseteq \calD$ such that $\calM'$ can be obtained from $\calM$ by starring any instance of $d_i \in D$ in $\calM$ and unstarring any instance of $d_i^*$ in $\calM$.
Since this corresponds to negating the $i$'th row of $\vM_\calT$, which negates the $i$'th entry of the vector $\vd = \vM_\calT \vvc$, this gives the following observation.

\begin{observation} \label{obs:Apos}
Let $\calT=(\calD,\calM)$ be a TBN and $\vvc\in\N^\calM$ a monomer collection.
There exists a relabeling $\calT'$ of $\calT$ so that $\vM_{\calT'} \vvc \geq 0$.
\end{observation}

Combining Observations~\ref{obs:excess-domains-same-across-saturated-configs} and~\ref{obs:Apos} results in the following observation,
which essentially states that for any given monomer collection $\vvc$,
we may assume without loss of generality that domains unbound in saturated configurations $\alpha\in[\vvc]$ are all primary domain types.

\begin{observation} \label{obs:Apos-and-excess-saturated}
  Let $\calT=(\calD,\calM)$ be a TBN and $\vvc\in\N^\calM$ a monomer collection.
  There exists a relabeling $\calT'$ of $\calT$ so that,
  letting $\vd = \vM_{\calT'} \vvc$,
  for all configurations $\alpha \in [\vvc]$,
  $\alpha$ is saturated if and only if, 
  for all $i \in \ints{d}$,
  $\vd(i) \in \N$ is the number of unbound primary domains of type $d_i \in \calD$ in $\alpha$.
\end{observation}

The following lemma is a key technical tool for showing that a polymer is not stable (or equivalently that a stable configuration has entropy greater than~1 and therefore cannot be a single polymer).
It generalizes the idea shown in Fig.~\ref{fig:long-polymer-breaking} that if one can find a monomer subcollection $\vvc_1$ in a larger collection $\vvc$,
and $\vvc_1$ has a saturated configuration with \emph{no} bonds left unbound,
then one can create a saturated configuration $\gamma \in [\vvc]$ with no bonds between $\vvc_1$ and $\vvc - \vvc_1$.
(Thus~$\gamma$ has at least two polymers.)

More generally,
given a monomer collection $\vvc$ with at least as many $d_i$ as $d_i^*$ domains (under appropriate relabeling this holds for each $i$ by Observation~\ref{obs:Apos}),
if we can partition $\vvc$ into subcollections $\vvc_1$ and $\vvc_2$,
and each of them \emph{also} has at least as many $d_i$ as $d_i^*$ domains for each $i \in \ints{d}$,
then every stable configuration $\alpha \in [\vvc]_\Box$ has at least two polymers,
since there is a saturated configuration of $\vvc$ in which there are no bonds between $\vvc_1$ and $\vvc_2$.\footnote{Observations~\ref{obs:excess-domains-same-across-saturated-configs},~\ref{obs:Apos}, and~\ref{obs:Apos-and-excess-saturated} are not really \emph{necessary} for our technique, but simplify the description of the conditions under which $\vvc_1$ and $\vvc_2$ would be saturated:
specifically, that if $\vd = \vM_\calT \vvc$ is in the nonnegative orthant, then so are $\vd_1 = \vM_\calT \vvc_1$ and $\vd_2 = \vM_\calT \vvc_2$.
If we did not use relabeling
(thus could not guarantee that $\vd$ is in the nonnegative orthant)
then the requisite condition to apply Lemma~\ref{lem:split} would be that $\vd$, $\vd_1$, and $\vd_2$ all occupy the same orthant;
i.e., for all $i \in \ints{d}$, if any of $\vd(i)$, $\vd_1(i)$, or $\vd_2(i)$ are negative, then the other two are not positive.}

\begin{lemma} \label{lem:split}
Let $\calT=(\calD,\calM)$ be a TBN,
let $\vvc \in \N^\calM$ be a monomer collection of $\calT$ such that $\vM_\calT \vvc \geq \vec{0}$,
and let $\alpha \in [\vvc]_\Box$ be a stable configuration.
If there exist nonempty subcollections $\vvc_1,\vvc_2\in\N^\calM$ where
1) $\vvc_1 + \vvc_2 = \vvc$ and
2) $\vM_\calT \vvc_1 \geq 0$ and $\vM_\calT \vvc_2 \geq 0$,
then $S(\alpha) > 1$.
\end{lemma}

\opt{arxiv_version}{The proof of Lemma~\ref{lem:split} appears in Appendix~\ref{app:split}.}
\opt{LNCS_version}{The proof of Lemma~\ref{lem:split} appears in the full version of this paper.}

\subsection{Exponential upper bound on polymer size}

We now show a converse to Theorem~\ref{thm:polymer-lower-bound-tree},
namely Theorem~\ref{thm:polymer-upper-bound}, showing that stable polymers have size at most exponential in the number of domain and monomer types.
The proof of Theorem~\ref{thm:polymer-upper-bound} closely follows Papadimitriou's proof that integer programming is contained in $\NP$~\cite{papadimitriou1981complexity}.
That proof shows,
for any linear system $\vA \vx = \vb$,
where $\vA$ is a given $n \times m$ integer matrix, $\vb \in \Z^n$ is a given integer vector,
and $\vx$ represents the $m$ unknowns,
that if the system has a solution $\vx \in \N^m$, then it has a ``small'' solution $\vx' \in \N^m$.
``Small'' means that $\amax(\vx')$ is at most exponential in $n+m+\amax(\vA)+\amax(\vb)$.
The technique of~\cite{papadimitriou1981complexity} proceeds by showing that any sufficiently large solution $\vx \in \N^m \setminus \{\vec{0}\}$ can be split into two vectors $\vx_1,\vx_2\in\N^m \setminus \{\vec{0}\}$ such that $\vx_1+\vx_2=\vx$,
where $\vA \vx_1 = \vec{0}$, so $\vx_2$ is also a solution:
$\vA \vx_2 = \vA (\vx-\vx_1) = \vA \vx - \vA \vx_1 = \vA \vx = \vb$.
This is useful because $\vx_1$ and $\vx_2$ satisfy the hypothesis of Lemma~\ref{lem:split},
which tells us that all stable configurations $\alpha \in [\vx]$ obey $S(\alpha)>1$,
so any single-polymer configuration of $\vx$ is not stable.

We include the full proof for three reasons:
1) self-containment,
2) it requires a bit of care to convert our inequality $\vA \vx \geq \vec{0}$ into an equality as needed for the technique,\footnote{In particular, the proof of~\cite{papadimitriou1981complexity} upper bounds the size of $\vx$ in terms of the entries of both $\vA$ and $\vb$. However, the na\"{i}ve way to solve a linear inequality $\vA \vx \geq \vec{0}$ using an equality, by introducing slack variables $\vb$ and asking for solutions $\vx\in\N^m$, $\vb\in\N^n$ such that $\vA \vx = \vb$, allows for the possibility that $\|\vb\|$ is very large compared to $\|\vA\|$, in which case upper bounding $\|\vx\|$ in terms of both $\vA$ and $\vb$ does not help to bound $\|\vx\|$ in terms of $\vA$ alone.} 
and
3) although the proof of~\cite{papadimitriou1981complexity} is sufficiently detailed to prove our theorem, the statement of the theorem in~\cite{papadimitriou1981complexity} hides the details about splitting the vector, which are crucial to obtaining our result.

We require the following discrete variant of Farkas' Lemma, also proven in~\cite{papadimitriou1981complexity}.

\newcommand{\sfT}{\mathsf{T}}

\begin{lemma}[\cite{papadimitriou1981complexity}]\label{lem:farkas}
  Let $a,d,l\in\Z^+$, $\vv_1,\ldots,\vv_l \in \{0,\pm 1,\ldots,\pm a\}^d$, 
  and $K = (ad)^{d+1}$.
  Then exactly one of the following statements holds:
  \begin{enumerate}
      \item\label{lem:farkas:balanced} There exist $l$ integers $n_1,\ldots,n_l \in \{0,1,\ldots,K\}$, not all $0$, such that \newline$\sum_{j=1}^l n_j \vv_j = \vec{0}.$

      \item\label{lem:farkas:one-side} There exists a vector $\vh \in \{0,\pm 1,\ldots,\pm K\}^d$ such that, for all $j \in \ints{l}$, $\vh^\sfT \cdot \vv_j \geq 1$.
  \end{enumerate}
\end{lemma}

Intuitively, statement~\eqref{lem:farkas:balanced} of Lemma~\ref{lem:farkas} states that the vectors can be added to get $\vec{0}$ (they are ``directions of balanced forces''~\cite{papadimitriou1981complexity}).
This is false if and only if statement~\eqref{lem:farkas:one-side} holds:
the vectors all lie on one side of some hyperplane, whose orthogonal vector $\vh$ would then have positive dot product with each of the vectors $\vv_j$ (thus adding any of them would move positively in the direction $\vh$ and could never cancel to get $\vec{0}$).

Intuitively, Theorem~\ref{thm:polymer-upper-bound} states that the size of polymers in stable configurations is upper bounded by a function which is exponential in $d$.  We prove this by first defining a constant $K$ which is exponential in $d$.  If each of the $m$ individual monomer counts is less than $K$, then we are done since no polymer in the configuration can have size bigger than $mK$.  If some of the monomer counts are greater than $K$ (call these \emph{large-count monomers}), we consider two cases.  

For the first case, we consider the scenario where the vectors which describe the monomer types with large monomer counts are such that they can ``balance'' each other out with relatively small linear combination coefficients.  If this is the case, then we can make a saturated subconfiguration which has at least one polymer using these small linear combination coefficients and large-count monomer types since the domains and codomains completely ``balance'' each other out.  We can then use the rest of the counts of the configuration to make another saturated subconfiguration which has at least one polymer.  This is shown mathematically by applying Lemma~\ref{lem:farkas} to show that the monomer counts in the polymer can be split to find a configuration consisting of two separate saturated polymers.  This means that there is a saturated configuration that has at least two polymers which contradicts the assumption $\alpha$ is a single stable polymer.

If there exist no such linear combination to ``balance out'' out the vectors describing the large-count monomers, then Lemma~\ref{lem:farkas} tells us all of these vectors lie on the same side of some hyperplane.  In this case, we show that counts of the small-count monomers play a role in bounding the counts of the large-count monomers.  Intuitively, if all of the vectors describing the large-count monomers lie on the same side of some hyperplane, they are missing domains and codomains which will allow them to bind together.  The domains and codomains they need in order to bind together, then must be found on the small-count monomer.  Consequently, this means the size of polymers will be bound by the counts of small-count monomers (which is exponential in $K$). \opt{LNCS_version}{The proof appears in the full version of this paper.}

\opt{arxiv_version}{The proof of the following theorem is in Appendix~\ref{app:proof-polymer-upper-bound}.}

\begin{theorem}\label{thm:polymer-upper-bound}
Let $\calT=(\calD,\calM)$ be a TBN with $d=|\calD|$ and $m=|\calM|$.
Let $a=\max\limits_{\vm \in \calM, d_i \in \calD\cup\calD^*} \vm(d_i)$ be the maximum count of any domain in any monomer.
Then all polymers of every stable configuration $\alpha$ of $\calT$ have size at most
$2 (m+d) (ad)^{2d+3}.$
\end{theorem}

{\small
\bibliographystyle{plain}
\bibliography{tam}
}


\newpage
\appendix

\section{Proof of Lemma~\ref{lem:translator}}
\label{app:translator}

\begin{proof}
  By structural induction on cascade depth $k$.
  Consider as a minimal element a saturated configuration $\gamma_1$ having output $1$, for a translator cascade of depth $k=1$, redundancy $n$, and without input.
  By assumption, the terminator monomer $\{x_2,...,x_2\}$ is free.
  To saturate the $n$ codomains of type $x_{2}^*$, the $\{x_2^*,...,x_2^*\}$ monomer and the $n$ $\{x_1^*,x_2\}$ monomers must be in the same polymer.
  To saturate the $n$ codomains of type $x_1^*$, the $n$ $\{x_1\}$ monomers must also be in the same polymer containing the $n$ $\{x_1^*,x_2\}$ monomers.
  There are therefore two polymers containing the following monomers:
  (1) the terminator monomer $\{x_2,...,x_2\}$, and (2) every other monomer.
  Thus, $S(\gamma_1)=2$.

  Assume that if $\gamma_k$ is a saturated configuration having output $1$ for a translator cascade of depth $k$, redundancy $n$, and without input, then $S(\gamma_k)=n(k-1)+2$.

  Consider a saturated configuration $\gamma_{k+1}$ having output $1$, for a translator cascade of depth $k+1$, redundancy $n$, and without input.
  Let $L_i=n \times \{x_i*,x_{i+1}\} \cup n\times \{x_{i}\}$ be the gate monomers of each layer $i$; $1 \leq i \leq k+1$.
  We first modify $\gamma_{k+1}$ into a saturated configuration $\gamma'_{k+1}$ with output $1$, such that there are no bonds between monomers in $L_1$ and $L_2$.
  The only possible bonds between monomers in $L_1$ and $L_2$ is between a monomer $\{x_1^*, x_2\} \in L_1$ and a monomer $\{x_2^*, x_3\} \in L_2$.
  Let $p$ be the number of bonds between these two types of monomers.
  If $p=0$, we are done.
  Otherwise, we note that there must be $p$ free $\{x_2\}$ monomers in $\gamma_{k+1}$.
  Let $\gamma'_{k+1}$ be the configuration where the $p$ bonds between $L_1$ and $L_2$ monomers in $\gamma_{k+1}$ are replaced by $p$ new bonds between $\{x_2\}$ and $\{x_2^*,x_3\}$ monomers, both from $L_2$.
  Thus $\gamma'_{k+1}$ remains saturated and can be partitioned into two saturated sub-configurations:
  $\gamma'_{L_1}$ containing the monomers from $L_1$ and $\gamma'_{R}$ containing the remainder.
  Since there are no bonds between the sub-configurations then $S(\gamma'_{k+1}) = S(\gamma'_{L_1}) + S(\gamma'_{R})$.

  First, note that the $2n$ monomers from $L_1$ can only form a single saturated configuration containing $n$ polymers, with each polymer being an $\{x_1\}$ monomer bound to a $\{x_1^*, x_2\}$ monomer.
  Thus, $S(\gamma'_{L_1})=n$.
  Second, note that $\gamma'_R$ is a saturated configuration having output $1$ for a translator cascade of depth $k$, redundancy $n$, and without input.
  Thus, $S(\gamma'_R)=n(k-1)+2$ by the inductive assumption.
  Therefore, $S(\gamma'_{k+1}) = nk + 2$ establishing the claim.
\end{proof}

\section{Proof of Lemma~\ref{lem:split}}
\label{app:split}

\begin{proof}
Let $\vd = \vM_\calT \vvc$, $\vd_1 = \vM_\calT \vvc_1$, and $\vd_2 = \vM_\calT \vvc_2$.
For $i \in \ints{d}$, $\vd(i) \geq 0$ is the number of unbound domains of type $d_i \in \calD$ in every saturated configuration $\alpha \in [\vvc]$ (and no $d_i^*$ domains are unbound in $\alpha$).
Let $\alpha_1\in[\vvc_1]$ and $\alpha_2\in[\vvc_2]$ be saturated configurations.
Define the configuration $\alpha\in[\vvc]$ by $\alpha = \alpha_1 \cup \alpha_2$;
note that $S(\alpha) > 1$ since there are no bonds between $\alpha_1$ and $\alpha_2$.
Excess domains of each type in $\alpha$ are given by the vector
$\vd_1 + \vd_2 = \vM_\calT \vvc_1 + \vM_\calT \vvc_2 = \vM_\calT (\vvc_1 + \vvc_2) = \vM_\calT \vvc = \vd$.
Thus $\alpha$ is saturated by Observation~\ref{obs:excess-domains-same-across-saturated-configs}.
Since $S(\alpha) > 1$, every stable configuration of $[\vvc]$ has at least two polymers as well.
\end{proof}

\section{Upper bound on the size of a stable polymer if it is a tree}
\label{app:acyclic-self-assembly}

\begin{figure}[t]
\centering
\includegraphics[width=\textwidth]{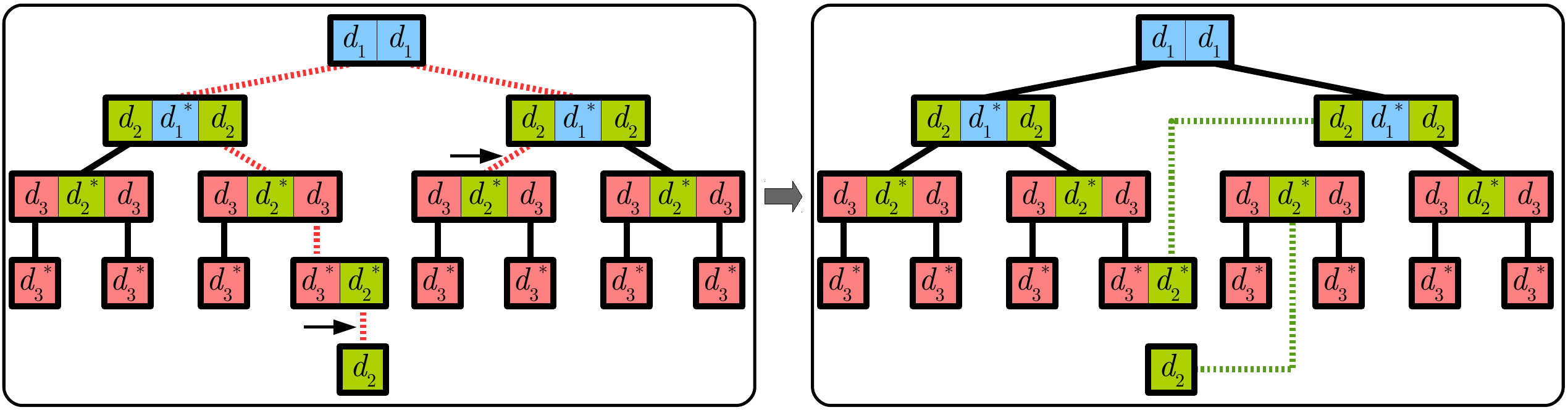}
\caption{
{\bf left:}
Example showing why a path (dashed red edges) in an acyclic monomer binding graph that repeats the same primary domain/codomain pair \emph{in the same order} implies the configuration is not stable.
This configuration is similar to Fig.~\ref{fig:tree-polymer}, but by extending the depth of the tree without creating a new domain type, we are forced to reuse the same domain type three times along a path, so two occurrences as one moves along the path (marked with black arrows) have the same primary/codomain ordering.
{\bf right:}
Exchanging the bonds between these occurrences (dashed green edges) disconnects the graph---increasing entropy---while keeping the configuration saturated.
}
\label{fig:tree-upper-bound}
\end{figure}

\begin{proposition}
  Let $\calT=(\calD,\calM)$ be a TBN,
  $d=|\calD|$,
  and $l = \max\limits_{\vm \in \calM} \|\vm\|_1$ be the maximum number of domains in any monomer type.
  Then any stable polymer $\rho$ of $\calT$ such that $G_\rho$ is acyclic has at most $1 + l \cdot \frac{(l-1)^{d}-l+1}{l^2-3l+2}$ monomers.
\end{proposition}

\begin{proof}
Fig.~\ref{fig:tree-upper-bound} shows an example.
A path in $G_\rho$ traverses an edge by either moving from a monomer with a codomain $d_i^*$ to a monomer with a primary domain $d_i$, or vice versa.
Suppose there is a simple path in $G_\rho$ of length $2d+1$.
Then by the pigeonhole principle, some domain $d_i$ is traversed twice with the same primary/codomain ordering.
Then by binding the first instance of $d_i$ to the second instance of $d_i^*$ and the second instance of $d_i$ to the first instance of $d_i^*$, we introduce a cycle.
However, since the number of edges has not changed, the new graph must be disconnected,
i.e., split into at least two polymers.
But since this new configuration is saturated, the original could not have been stable.
Therefore the diameter of any acyclic stable monomer binding graph is at most $2d$.

Note that $l$ upper bounds the degree of $G_\rho$.
Any degree-$l$ acyclic graph of diameter $\leq 2{d}$ has at most $1 + l \cdot \sum_{j=1}^{d-1} (l-1)^{j-1} = 1 + l \cdot \frac{(l-1)^{d}-l+1}{l^2-3l+2}$ nodes:
the worst case is a rooted complete $(l-1)$-ary tree of depth $d$ (so diameter $2d$), whose root is the only node that can have up to $l$ children since it has no parent. 
\end{proof}

\section{Proof of Theorem~\ref{thm:polymer-upper-bound}}
\label{app:proof-polymer-upper-bound}

\begin{proof}
It suffices to prove the theorem in the special case where $\alpha$ is itself a single stable polymer.
Let $\vvc$ be the monomer collection of $\alpha$.
By Observation~\ref{obs:Apos} we can assume without loss of generality that $\vM_\calT \vvc \geq \vec{0}$. 

To turn this inequality into an equality, we introduce slack variables.
Let $\vI_d$ be the $d \times d$ identity matrix.
Let $\vM' = [\vM_\calT | -\vI_d]$ ($\vM_\calT$ concatenated horizontally with $-\vI_d$); note $\vM'$ has dimension $d \times (m+d)$.
Let $\vd = \vM_\calT \vvc \in \N^d$; by Observation~\ref{obs:Apos-and-excess-saturated} $\vd(i)$ is the number of unbound primary domains $d_i \in \calD$ in every saturated configuration in $[\vvc]$ (and no such configuration has any unbound codomains).
Concatenating $\vvc$ with $\vd$ to obtain $\vvc' \in \N^{m+d}$, we have $\vM' \vvc' = \vec{0}$.
Let $\vm'_1,\ldots,\vm'_{m+d} \in \N^d$ denote the columns of $\vM'$.

Intuitively, $\vM'$ can be thought of as the monomer matrix for a modified TBN that has all the monomer types of $\calM$,
and in addition, for each $i \in \ints{d}$, also has monomer type $\vm'_{m+i} = \{d_i^*\}$.
For each $i \in \ints{d}$, $\vvc'(m+i)$ represents how many monomers of type $\vm'_{m+i}$ should be added to a saturated configuration of $[\vvc]$ to bind to every unbound $d_i$ primary domain, so that the resulting configuration has no unbound domains (i.e., $\vM' \vvc' = \vec{0}$).

Let $K=(a d)^{d+1}$.
If $\max(\vvc) < K$ (i.e., all monomers have count at most $K$),
then $\|\vvc\| < m K = m (a d)^{d+1} < 2 (m+d) (ad)^{2d+3}$, and we are done.
Otherwise, let $J = \{j_1, j_2,\ldots,j_l\} = \{ j \in\{1,\ldots,m+d\} \mid \vvc'(j) \geq K \}$ denote the set of indices of components of $\vvc'$ that are at least $K$ (we call these \emph{large-count monomer types} below),
and let $\overline{J} = \{1,\ldots,m+d\} \setminus J$ be all the other indices.
Let $\vv_1,\vv_2,\ldots,\vv_l$ denote the corresponding columns of $\vM'$, i.e., $\vv_i = \vm'_{j_i}$.

We now consider two cases.

\begin{enumerate}
\item 
    There exist integers $n_1, \ldots, n_{l}$ between $0$ and $K$, not all $0$, such that
    $\sum_{j = 1}^{l} n_j \vv_j= \vec{0}$.
    Define $\vvc'_1 \in \mathbb{N}^{m+d}$ by $\vvc'_1(j_i) = n_i$ for $i \in \{1,\ldots,l\}$ and $\vvc'_1(j)=0$ for all $j \in \overline{J}$.
    In other words, $\vvc'_1$ has $n_i$ copies of the $i$'th large-count monomer type and none of the others.
    Define $\vvc'_2 = \vvc' - \vvc'_1$.  Note that by the way $\vvc'_1$ was constructed $\vvc'_2 \geq \vec{0}$.
    We have $\vM' \vvc'_1 = \sum_{j = 1}^l n_j \vv_j = \vec{0}$ by definition.
    From this it follows that
    $\vM' \vvc'_2 = \vM' (\vvc'-\vvc'_1) = \vM' \vvc' - \vM' \vvc'_1 = \vec{0}$.

    Define $\vvc_1=(\vvc'_1(1), \vvc'_1(2), \ldots, \vvc'_1(m))$ and $\vvc_2=(\vvc'_2(1), \vvc'_2(2), \ldots, \vvc'_2(m))$.
    Since $\vvc'_1+\vvc'_2=\vvc'$, we have $\vvc_1+\vvc_2=\vvc$.
    The last $d$ columns of $\vM'$, i.e., $\vm'_{m+1}, \ldots, \vm'_{m+d}$, are all negative unit vectors,
    and $\vvc'(m+1), \ldots, \vvc'(m+d) \geq 0$.
    Thus, since $\vM' \vvc_1' = \vec{0}$,
    we have that $\vM_\calT \vvc_1 \geq \vec{0}$.
    A similar argument establishes that $\vM_\calT \vvc_2 \geq \vec{0}$.
    Since $\vvc_1, \vvc_2 \in \mathbb{N}^{\calM}$ and $\vvc_1 + \vvc_2 = \vvc$, by Lemma~\ref{lem:split}, $S(\alpha)>1$, a contradiction since we assumed $\alpha$ is a single polymer.

\item 
    No such integers exist.
    We show that this implies that the large-count monomer types can't be \emph{too} large.
    By Lemma~\ref{lem:farkas} there is a vector $\vh \in \{0, \pm 1, \ldots, \pm K\}^d$ such that
    $\vh^\sfT \cdot \vv_j \geq 1$ for all $j \in \{1, \ldots, l\}$;
    put another way, $\vh^\sfT \cdot \vm'_j \geq 1$ for all $j \in J$.
    Taking the dot product of $\vh^\sfT$ with both sides of the equation $\vM' \vvc' = \vec{0}$ gives
    $\vh^\sfT \cdot \vM' \vvc' = \vh^\sfT \cdot \vec{0} = 0$, i.e.,
    $\sum_{j=1}^{m+d} (\vh^\sfT \cdot \vm'_j) \vvc'(j) = 0$.
    Moving the sum elements corresponding to $\overline{J}$ to the other side of this equation,
    $$
      \sum_{j\in J} (\vh^\sfT \cdot \vm'_j) \vvc'(j)
      =
      - \sum_{j \in \overline{J}} (\vh^\sfT \cdot \vm'_j) \vvc'(j)
    $$
    Since $1 \leq \vh^\sfT \cdot \vm'_j$ for all $j \in J$,
    \begin{align*}
      \sum_{j \in J} \vvc'(j)
      \leq 
      \sum_{j \in J} (\vh^\sfT \cdot \vm'_j) \vvc'(j)
      = 
      -\sum_{j \in \overline{J}} (\vh^\sfT \cdot \vm'_j) \vvc'(j)
      < 
      (m+d) d K^2 a,
    \end{align*}
    where the last inequality follows from the fact that
    there are $|\overline{J}| < m+d$ entries in the sum,
    $d$ terms in each dot product $\vh^\sfT \cdot \vm'_j$,
    $0 \leq \vvc'(j) < K$ for all $j \in \overline{J}$,
    $\amax(\vh) \leq K$,
    and
    $\amax(\vM') \leq a$.
    Also note
    $\sum_{j\in\overline{J}} \vvc'(j) < (m+d) K.$
    To complete the proof, observe that since $\vvc'$ is $\vvc$ with extra concatenated nonnegative components,
    \begin{align*}
      \|\vvc\| \leq \|\vvc'\|
      & =
      \sum_{j \in J} \vvc'(j) + \sum_{j \in \overline{J}} \vvc'(j)
      < 
      (m+d) d K^2 a + (m+d) K
      \\ & <
      2 (m+d) d K^2 a
      = 
      2 (m+d) (ad)^{2d+3}. \qedhere
    \end{align*}
\end{enumerate}

\end{proof}

\section{Self-assembly: kinetics versus thermodynamics}
\label{sec:kineticexamples}

In this section we compare our thermodynamic model (TBNs) with a well-studied kinetic model of self-assembly called the abstract Tile Assembly Model (aTAM).  We first show that it is not necessarily the case that the output (terminal) assembly of an aTAM system is thermodynamically stable when interpreted as a polymer in the TBN model. Specifically we show this for an aTAM binary counter system. 
Then, we propose a new aTAM system which assembles a binary counter that is also  stable polymer in the TBN model. This shows that it is possible, at least for a binary counter, to have systems that are tile/monomer efficient and  that self-assemble to an intended target structure in both a kinetic and thermodynamic sense.

\subsection{Interpreting aTAM assemblies as TBN polymers}\label{sec:aTAMtoTBN}
The  fundamental components in the aTAM are non-rotatable unit-sized 2D squares with a single ``glue'' on each of their four sides.  A glue consists of a label (which is usually a string or color) and a non-negative integer value called the strength.  The glues of two tiles can bind provided that their labels and strengths match.  Growth in the aTAM occurs on the integer square lattice of points ($\mathbb{Z}^2$) and begins with a connected initial configuration of tiles, called a seed, and then proceeds in an asynchronous manner with one tile attaching at a time.  A tile can attach to the existing assembly provided that the combined strength of all the glues with which it binds is greater than or equal to a parameter of the system which we call the temperature. An assembly has an associated binding graph: an undirected graph on $\mathbb{Z}^2$ with nodes corresponding to tiles in the assembly and edges corresponding to bound glues. An assembly to which no more tiles can attach is said to be terminal, Fig.~\ref{fig:aTAM-counter}(a) shows an example. 
Formally, an aTAM system is defined to be a triple $(T,\sigma,\tau)$ where $T$ is a finite set of tiles, $\sigma\in T$ is the seed tile and $\tau\in\mathbb{Z}^+$ is the temperature. 
 Comprehensive informal and formal definitions of the aTAM can be found in~\cite{PatitzSurveyJournal}.

We can interpret any aTAM assembly as a TBN polymer in the following way. Each aTAM tile is mapped to a TBN monomer and each aTAM glue on a tile side is mapped to a TBN domain in a monomer. The binding graph of an aTAM assembly is mapped to a corresponding TBN configuration (which, due  geometry, is a graph and not a multigraph).  
 Notion in the aTAM tiles having a square shape, there being a seed tile, temperature, and binding occurring on the integer lattice have no corresponding notions in the TBN.

\subsection{A tile-based binary counter that is not a thermodynamically stable polymer}
In this subsection we show that it is not necessarily the case that the terminal assembly of an aTAM system is thermodynamically stable when interpreted as a polymer in the TBN model. 
To get an immediate intuitive sense of a difference in the two models simply observe that a monomer in the TBN model, can ``self-saturate'' meaning that if it contains two domains which are complementary, then those two domains can bind to each other as opposed to binding to domains of other monomers.\footnote{Experimentally, for example, such a situation can occur when using the DNA single stranded tile motif~\cite{reif08yin}.} Even worse, in a TBN system there is no geometric requirement on the arrangement of bonds, so polymers can form into arbitrary networks. 
However, in the aTAM square tiles can never self-bind, nor can they bind to tiles that are not adjacent on the 2D square lattice.

Part (a) of Fig.~\ref{fig:aTAM-counter} shows the unique terminal assembly of an aTAM system which assembles a binary counter that counts to $16$.  We briefly describe the system here.  The system has a tile set $T$ which is the union of all tiles shown in Fig.~\ref{fig:aTAM-counter}(a), a seed tile $\sigma$ which is the bottom leftmost tile, and works at temperature 2.  
Strength 1 and 2 glues are respectively denoted as thin and thick coloured lines between adjacent tiles. 
The tile set $T$ contains four ``hard-coded'' tiles with strength 2 glues so that the leftmost (or ``first'')  column of the counter can grow directly from the seed at temperature~2.  Notice in the assembly that the strength~2 glues which initiate the growth of new columns alternate between occurring on the east of the bottom tile of the column and occurring on the east of some other tile in the column. 
When they occur on the bottom, growth continues upwards with a carry $c$ being propagated, causing bits to be flipped from west to east, until a $0$ is met. 
If the strength 2 glues initiates a column from elsewhere, there are tile types in $T$ so that the values of glues are simply copied to the east except for the least significant (i.e.\ bottom) bit which is incremented from $0$ to $1$.  
 Growth of the counter is halted when the red strength 2 glue is exposed which causes a column of tiles to bind which do not have any output glues.


We can interpret the assembly shown in Fig.~\ref{fig:aTAM-counter}(a) as a TBN polymer $\alpha$ by viewing the square tiles as monomers and their glues as domains, as described  in Section~\ref{sec:aTAMtoTBN}. Observe that the TBN configuration  $\beta$ in Fig.~\ref{fig:aTAM-counter}(b) has the same multiset of monomers as~$\alpha$.
Furthermore, one of the tile-monomers of $\beta$ is self-saturated (each of its domains is bound to another domain on the same monomer). Finally, observe that $H(\alpha) = H(\beta)$ (both have the same number of bound domains) but $S(\alpha) < S(\beta)$ ($\beta$ has one more polymer than $\alpha$). Hence  $\alpha$ is not a stable configuration. 
This yields the straightforward conclusion that a terminal assembly producible in the aTAM is not necessarily a thermodynamically stable polymer when interpreted in the TBN model.


\begin{figure}[t]
\centering
\includegraphics[width=\textwidth]{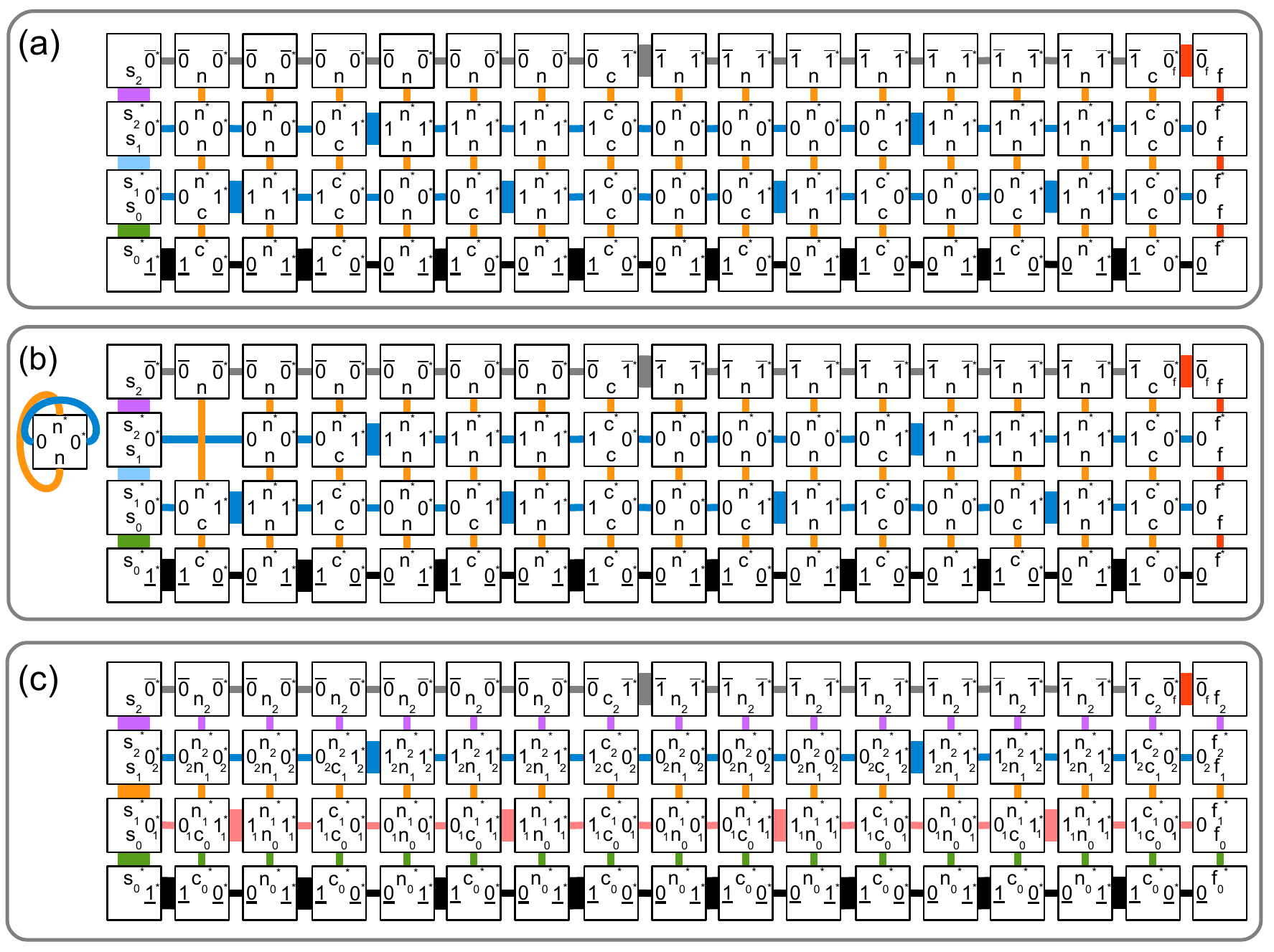}
\caption{\textbf{(a)} ~A binary counter in the abstract Tile Assembly Model (aTAM) that self-assembles a $k \times (2^k+1)$ rectangle from $k + O(1)$ square tiles, where $k=4$. Thick lines between tiles are ``strength 2'' domains and thin lines are ``strength 1'' domains.  Colors of glues between tiles help to indicate ``positional specificity''.
(For aTAM glues we ignore all ``$\ast$'' annotations in the figure.) 
This assembly can be interpreted as a TBN polymer (where tiles become monomers, and glues become domains). However, the resulting TBN polymer is not stable, since \textbf{(b)} shows another stable TBN configuration using the same set of monomers, but with greater entropy since it is composed of two polymers (i.e.\  monomers can ``self-saturate'').
\textbf{(c)}~An aTAM counter using $O(k)$ tiles, that is also a stable polymer when interpreted in the TBN model. 
Compared to part (a) and (b), in (c)  each monomer encode its height $i$ (i.e., the significance of the relevant bit), by having domains and codomains subscripted with $i$ or $i-1$. In the proof of its stability, this polymer is denoted $\gamma$.}
\label{fig:aTAM-counter}
\end{figure}

\subsection{A thermodynamically stable binary counter}
Part (c) of Fig.~\ref{fig:aTAM-counter} shows a binary counter that is a terminal assembly producible by an aTAM system. 
Let $\gamma$ be the interpretation of that assembly as a TBN polymer.  
This section gives a proof of Lemma~\ref{lem:gammaIsStable}, but first we give a little intuition. 

 To create the tile assembly system that builds the terminal assembly in Fig.~\ref{fig:aTAM-counter}(c), the tile assembly system in Fig.~\ref{fig:aTAM-counter}(a) is modified so that the labels of the glues encode the natural number ``height'' at which they appear. 
Intuitively we can already see an obvious difference between $\beta$ and $\gamma$ since in  $\gamma$  the ``top'' and ``bottom'' domain of each monomer is distinct, hence it is not possible for any monomer in $\gamma$ to bind only to itself; it must bind to other  monomers in a stable configuration. However, this observation alone is not sufficient to show that the single-polymer configuration $\gamma$ in Fig.~\ref{fig:aTAM-counter}(c) is stable. In the remainder of this section we argue that that is indeed the case. 

\begin{lemma}\label{lem:gammaIsStable}
$\gamma$ is a stable polymer in the TBN model. 
\end{lemma}

We define the \emph{seed column} to be the set of monomers in the leftmost column of $\gamma$ and 
and the \emph{end column} to be the set of monomers in the rightmost column.  
We also enumerate rows of $\gamma$ from bottom to top starting starting with $1$ for the bottom row.  
We refer to any monomer that contains a subscripted $c$ or $c^*$ glue in row $i \in \{1,2,3,4\}$  of $\gamma$  (counting from the bottom) as a $c^i$-monomer.  In addition, if a $c^i$-monomer contains a $c_{i-1}^*$ domain, we refer to it as a $c^i_\mathrm{n}$-monomer, otherwise we refer to it as a $c^i_\mathrm{s}$-monomer.  Furthermore, we refer to monomers in row $i$ which contain subscripted $n$ domains as $n^i$-monomers.

\newcommand{\containsSeed}{\delta}
To see that the configuration containing exactly $\gamma$ is thermodynamically stable, we will show that all stable configurations consist of one polymer with all domains bound.  

Let $C$ be {\em any} stable configuration consisting of the monomers shown  Fig.~\ref{fig:aTAM-counter}(c). Let $\containsSeed$ be the unique polymer in $C$ which contains the  {\em seed monomer} (i.e. the bottom-left tile of Fig.~\ref{fig:aTAM-counter}(c) interpreted as a monomer).  We will show that $\containsSeed$ must in fact contain all monomers in the polymer shown in Fig.~\ref{fig:aTAM-counter}(c), which in turn implies Lemma~\ref{lem:gammaIsStable}.

Let row $i$ of $\containsSeed$ be the set of all monomers in $\containsSeed$ that appear on row  $i$ of $\gamma$. Likewise left column $i$ of $\containsSeed$ be the set of all monomers in~$\containsSeed$ that appear on column~$i$ of~$\gamma$.

The following observation states that $\containsSeed$ cannot contain any partial columns. 
\begin{claim} \label{obs:no-partials}
The rows of $\containsSeed$ all contain the same number of monomers.
\end{claim}
\begin{proof}
$\gamma$ does not contain any unbound domains, therefore by stability neither does $\containsSeed$. 
If two rows of $\containsSeed$ each had a different number of monomers then this would imply that some domain of $\containsSeed$ is unbound, which is impossible.   
 \end{proof}

The next observation also follows from the fact that $\containsSeed$ cannot contain any unbound domains.
\begin{claim} \label{obs:end-cols}
  The polymer $\containsSeed$ contains the seed column and end column. 
\end{claim}
\begin{proof}
First, note that since $\containsSeed$ contains the seed monomer by assumption, it must contain the entire seed column; otherwise some domain would be unbound because the domains in the seed column are unique and do not appear on any other monomers.   
To see that $\containsSeed$ must contain the end column we examine the bottom row of $\containsSeed$.  We have already established that $\containsSeed$ must contain the seed column.  Now, observe that the bottom monomer of the seed row contains a horizontal starred domain and all the inner monomers (the monomers which do not appear in the seed or end columns) contain exactly one starred horizontal domain and one unstarred horizontal domain.  Thus, if the bottom monomer in the end column was not in $\containsSeed$, there would be one extra unstarred horizontal domain in the bottom row of $\containsSeed$ which means $\containsSeed$ would contain unbound domains. And, for the same reason $\containsSeed$ must contain the entire seed column, it must contain the entire end column.
\end{proof}

Since $\containsSeed$ contains the seed and end columns and the top row of $\containsSeed$ cannot contain any unbound domains, we have the following:

\begin{observation} \label{obs:c4}
  The polymer $\containsSeed$ contains both $c^4$ monomers.  
\end{observation}

For $i \in \{1, 2, 3, 4\}$, we denote the number of $c^i$ monomers in $\containsSeed$ by $|c^i|$.  Likewise, we denote the number of $c^i_\mathrm{n}$ and $c^i_\mathrm{s}$ monomers in $\containsSeed$ by $|c^i_\mathrm{n}|$ and $|c^i_\mathrm{s}|$ respectively.

\begin{claim} \label{obs:cequal}
  Let $i \in \{2, 3\}$.  Then $|c^i_\mathrm{n}| = |c^i_\mathrm{s}|$.
\end{claim}
\begin{proof}
To see this, first note that the number of subscripted $0$ domains must equal the number of subscripted $0^*$ domains in row $i$  so that no domains are left unbound.  Denote the number of monomers which contain a subscripted $n$ domain in row $i$ by $|n^i|$. 
Observe that monomers on row $i$ that have both subscripted $n^*$ and  $n$ domains each have an equal number of subscripted $0$ and $0^*$ domains. Also, notice that for each $i$  the two monomers on row $i$ of the seed and end columns together (i.e.\ unioned) have an equal number of subscripted $0$ and $0^*$ domains.  Now observe that $c^i_s$ monomers have an extra subscripted $0$ domain and $c^i_\mathrm{n}$-monomers have an extra subscripted $0^*$ domain.  Also, we note that a row of monomers in row $i$ is made up completely of monomers from the seed and end columns, $c^i_\mathrm{n}$-monomers, $c^i_s$-monomers, and monomers that each contain a pair of subscripted $n, n^\ast$ domains.  Hence, in order for row $i$ to have an equal number of subscripted $0$ and $0^*$ domains, the number of $c^i_\mathrm{n}$ and $c^i_s$ monomers must be equal. \end{proof}

\newpage
\begin{claim} \label{obs:cdouble}
For $i > 2$, $|c^i| = \frac{|c^{i-1}|}{2}$ and $|c^2|=|c^1|$.
\end{claim}
\begin{proof}
In order for all domains in $\containsSeed$ to be bound, row $i$ must contain just enough $c$-monomers to bind with all the $c^{i-1}_\mathrm{n}$-monomers.  It immediately follows from this that $|c^2|=|c^1|$ since the $c$-monomers in row $1$  each contain the $c^*_0$ domain which in turn is bound to some $c_0$ domain which in turn are only found on row~2 in $c^2$-monomers. 

 Let $i > 2$. Claim~\ref{obs:cequal} states that the number of $c^{i-1}_\mathrm{n}$-monomers and $c^{i-1}_s$-monomers is equal. 
 Hence exactly half of the $c^{i-1}$ monomers, in fact the set of $c^{i-1}_\mathrm{n}$-monomers, 
 contain the $c^*_{i-2}$ domain that that can bind to a $c^i$ monomer.
 Also,  all $c^i$ monomers have a $c_{i-2}$ domain that binds to a $c^{i-1}$ monomer.
   Thus, the number of  $c$-monomers in row $i$ is exactly half the number $c$-monomers in row $i-1$.  In other words,  row $i$  contains exactly $\frac{|c^{i-1}|}{2}$ of $c^i$-monomers when $i > 2$.
\end{proof}

The next claim allows us to determine the number of columns contained in $\containsSeed$ based on the number of $c^1$-monomers contained in~$\containsSeed$.\footnote{Although we are not yet saying whether the number of $c^1$-monomers in $\containsSeed$ is the same as the number of $c^1$-monomers in $\gamma$}.

\begin{claim} \label{obs:c1}
The number of columns in $\containsSeed$ is $2+2|c^1|-1$.
\end{claim}
\begin{proof}
To see why Claim~\ref{obs:c1} holds, note that  Claim~\ref{obs:end-cols} tells us that $\containsSeed$ contains both the seed and end columns.  In addition, it follows from Claim~\ref{obs:no-partials} that for each $c^1$-monomer in $\containsSeed$, $\containsSeed$ must contain a full column.  In order for $\containsSeed$ not to have any exposed domains there must be an $n^1$-monomer between each $c^1$-monomer.  Once again, applying Claim~\ref{obs:no-partials}, for each of these $n^1$-monomers, $\containsSeed$ must contain a full column.  Consequently, it follows that $\containsSeed$ contains $2+2|c^1|-1$ columns.
\end{proof}

We are now ready to conclude this section with the proof of its main claim, Lemma~\ref{lem:gammaIsStable}. 
\begin{proof}[Proof of Lemma~\ref{lem:gammaIsStable}]
It follows from Observation~\ref{obs:c4} that $\containsSeed$  contains both $c^4$-monomers.  Claim~\ref{obs:cdouble} then implies that row $3$ has $4$ $c^3$-monomers.  Applying Claim~\ref{obs:cdouble} again, we have that $\containsSeed$  contains $8$ $c^2$-monomers.  By the kind of reasoning used in the proof of Claim~\ref{obs:cdouble} once more, we have that $\containsSeed$ contains $8$ $c^1$-monomers.  Using Claim~\ref{obs:c1}, we have that $\containsSeed$  contains $17$ columns and thus contains all the monomers of $\gamma$ (shown in Fig.~\ref{fig:aTAM-counter}(c)).  Consequently, any stable configuration has exactly one polymer, and since $\gamma$ is saturated, it is  stable.
\end{proof}

\end{document}